\documentclass[journal]{IEEEtran}

\usepackage[colorinlistoftodos,prependcaption,textsize=tiny]{todonotes}
\usepackage{ifpdf}
\usepackage{cite}
\usepackage{graphicx}
\usepackage{amssymb}
\usepackage{amsmath,amsfonts}
\usepackage{stmaryrd} 
\usepackage{xcolor}
\usepackage{algorithmic}
\usepackage{textcomp}
\usepackage{bbm}
\usepackage{stmaryrd} 
\usepackage{amsthm}
\usepackage{mathrsfs}
\usepackage{tikz}
\usepackage{subcaption}
\usepackage{inputenc}
\usepackage{array,multirow}
\usepackage{setspace}
\usetikzlibrary{fit,backgrounds}
\usetikzlibrary{automata,arrows,positioning,calc}
\definecolor{celadon}{rgb}{0.67, 0.82, 0.59}
\definecolor{cblue}{rgb}{0.6, 0.73, 0.89}
\theoremstyle{remark}
\newtheorem{thm}{Theorem}
\newtheorem{prop}{Proposition}
\newtheorem{cor}{Corollary}
\newtheorem{lem}{Lemma}
\newtheorem{defi}{Definition}
\newtheorem{rem}{Remark}
\newcommand{\norm}[1]{\left\lVert#1\right\rVert}
\begin{document}
	\title{Bounded Estimation over Finite-State~Channels: Relating Topological Entropy and Zero-Error~Capacity}
	
	\author{\IEEEauthorblockN{Amir Saberi, Farhad Farokhi and Girish N.~Nair}
		\thanks{The authors are with the Department of Electrical and Electronic Engineering, The University of Melbourne, VIC 3010, Australia  (e-mails: asaberi@student.unimelb.edu.au, \{ffarokhi, gnair\}@unimelb.edu.au). Aspects of this work were done while F. Farokhi was also affiliated with CSIRO's Data61. This work was supported by the Australian Research Council via Future Fellowship grant FT140100527. This paper was presented in part at the 2019 and 2020 IEEE International Symposia on Information Theory ~\cite{saberi2019state,saberi2020explicit}.}}
	
	\maketitle

\begin{abstract}
	We investigate state estimation of linear systems over channels having a finite state not known by the transmitter or receiver. We show that similar to memoryless channels, zero-error capacity is the right figure of merit for achieving bounded estimation errors. We then consider finite-state, worst-case versions of the common erasure and additive noise channels models, in which the noise is governed by a finite-state machine without any statistical structure. Upper and lower bounds on their zero-error capacities are derived, revealing a connection with the {\em topological entropy} of the channel dynamics. Separate necessary and sufficient conditions for bounded linear state estimation errors via such channels are obtained. These estimation conditions bring together the topological entropies of the linear system and the discrete channel.
\end{abstract}

\section{Introduction}
The communication channels that connect networked controlled systems make the traditional assumption that signals are continuously and perfectly available invalid. For example, data transmissions from wireless sensors are susceptible to noise, fading and interference from other transmitters. Often in communication systems, such effects are modelled by i.i.d. or Markov processes and performance is considered on average over many uses; e.g. in a normal phone call it is not important that every transmitted bit is received, and the average guaranteed performance is enough. In contrast, control systems often deal with safety-critical applications, where stability and performance must be guaranteed for every single use. As an example, in some applications, robots are required to avoid collisions especially operating alongside humans and, for that, they have to be localized with a bounded error \cite{haddadin2017robot}. Moreover, uncertainties such as disturbances and faults, often arise from mechanical and chemical components which may not exhibit i.i.d. randomness. \emph{Adversarial noise} is another uncertainty that may not be described by {\em a priori} known stochastic assumptions. Thus, it is often treated as a bounded unknown variable without statistical assumptions and its worst-case behaviours are considered \cite{teixeira2012attack,linsenmayer2017stabilization}.

As a major engineering problem, estimation over noisy channels has been studied extensively in the communications and networked control systems literature; see, e.g.~\cite{schenato2007foundations,matveev2009estimation,sukhavasi2016linear}, and references therein. This literature largely considers memoryless channels with known probabilistic models. The bounded data loss models used in~\cite{bernat2001weakly,xiong2007stabilization,linsenmayer2017stabilization} are notable departures from the stochastic approach, and use deterministic \textit{weakly hard real-time} constraints to model the loss process  in control and data networks. Furthermore, in low-latency communications where using long block codes is not feasible, deterministic, worst-case models have been proposed for coding \cite{badr2017layered,fong2019optimal}. In addition, in recent years, bounded, nonstochastic approximations have been employed to deal with the complexity of high-dimensional stochastic problems, e.g. in network information theory \cite{avestimehr2011wireless} and multidimensional stochastic optimization \cite{bandi2012tractable}. 

When estimating the state of a linear system via a noisy memoryless channel, it is known that the relevant figure of merit for achieving estimation errors that are almost-surely uniformly bounded over time is the {\em zero-error capacity}, $ C_0 $, of the channel, and not the ordinary capacity \cite{matveev2007shannon, Franceschetti2014,wiese2018secure}. 
The zero-error capacity is defined as the maximum block coding rate yielding exactly zero decoding errors at the receiver. For memoryless channels, the zero-error capacity depends on the graph properties of the channel, not on the values of non-zero transition probabilities~\cite{shannon1956zero}. Unfortunately, the zero-error capacity is zero for common memoryless stochastic channel models, e.g. binary symmetric and binary erasure channels~\cite{korner1998zero}. This means that when estimating unstable plants, the worst-case estimation errors grow unboundedly in the long term. Therefore, these channels are not suitable for safety- or mission-critical applications that must respect hard guarantees at all times. 

Channel models with memory can capture the case when noise patterns are correlated, instead of assuming i.i.d. noise in the channel such as intersymbol-interference. Moreover, these models better reflect communication situations in which network congestion or wireless fading can cause bursty error patterns that are difficult to model stochastically or that the underlying probabilities vary significantly over time~\cite{badr2017layered,badr2017perfecting,wang2018end}. Interestingly, channels with memory can have positive $ C_0 $ unlike memoryless erasure and symmetric channels. A large class of channels with memory can be captured by {\em finite-state channels}, where the transitions between the states of the channel govern the noise pattern. As there might be no probabilistic structure in these channels, worst-case scenarios are considered. Such models have received attention in the recent literature~\cite{nair2013nonstochastic, badr2017layered, wang2018end}, and are useful when probabilistic information  about the channel noise is not available or when the noise itself is not random, e.g., adversarial noises \cite{teixeira2012attack, zhu2019observer, zhu2019quasi}. %This is also true for achieving uniformly bounded estimation errors over nonstochastic memoryless channels~\cite{nair2013nonstochastic}. 
However, very few studies have considered the problem of estimation and control over channels with memory, e.g. stabilization over a special case of a continuous alphabet channel, namely a moving average Gaussian channel  is  studied in \cite{middleton2009feedback}.  This paper intends to fill this gap for discrete channels.

Most studies of finite-state channels are focused on finding the ordinary capacity in stochastic settings. As an example, the {\em Gilbert-Elliot} channel is a well-studied model for bursty error patterns; see \cite{goldsmith1996capacity}, \cite{han2015randomized}, and references therein. However, there is no general result for finite-state channels. In recent years, studies have been done towards determining the zero-error capacity of special channels with memory. In \cite{kovavcevic2017zero}, the zero-error capacity of special symbol shift channels is determined. In \cite{cao2018zero}, Cao \emph{et. al.} studied the zero-error capacity of binary channels with one memory. In \cite{zhao2010zero}, Zhao and Permuter introduced a dynamic programming formulation for computing the {\em zero-error feedback capacity} of a finite-state channel assuming the state of the channel is available at both the encoder and decoder. 

\subsection{Our contribution}
In this paper, we show that  the zero-error capacity of any finite-state channel coincides with the largest possible rate of {\em nonstochastic information} (introduced in \cite{nair2013nonstochastic}) across it. This result generalizes the conditions on bounded state estimation for linear systems from memoryless to finite-state channels, demonstrating the universality of these conditions over more general classes of communication channels. %Lower and upper bounds on the zero-error capacity of finite-state channels are derived, as well as the zero-error feedback capacity. 
It is worth noting that, in the case of probabilistic transitions, the zero-error capacity bounds presented in this paper are still valid since they do not depend on the transition probabilities. However, some examples of finite-state channels without stochastic assumptions are given that can be {\em bursty} channels \cite{fong2019optimal}, {\em sliding window} channels \cite{fong2018optimal}, and Gilbert-Elliot like channels \cite{badr2017layered}. In contrast to common stochastic channel models, the results presented here show that for the finite-state channel, both zero-error capacities can be strictly positive.

Next, we turn our focus to the state estimation problem over discrete erasure and additive noise channels governed by a finite-state machine. Such channels possess memory and thus fall outside the framework of~\cite{nair2013nonstochastic, matveev2007shannon,shannon1956zero}. Moreover, the input does not affect the noise process, making this problem different than \cite{cao2018zero}.
An upper bound on $C_0$ for finite-state erasure channels is derived by applying the dynamic programming equation in~\cite{zhao2010zero} which gives the exact value of zero-error feedback capacity, $ C_{0f} $. Novel bounds on $C_0$ are then derived in terms of the {\em topological entropy} of the channel state dynamics. The topological entropy is a metric to capture the asymptotic growth rate of uncertainty in a continuous dynamical system  evolving on a compact space, first introduced by Adler \emph{et. al.} \cite{adler1965topological}. We refer the reader to \cite{downarowicz2011entropy} and recent papers \cite{kawan2019metric,liberzon2017entropy} for detailed discussions. In addition, discrete topological entropy in {\em symbolic dynamics} is defined as the asymptotic growth rate of the number of possible  state sequences \cite{lind1995introduction}. 
%In the conference paper \cite{saberi2020explicit}, the exact value of the zero-error feedback capacity of finite-state additive noise channels based on the topological entropy of the channel is derived.
The upper bound derived for $ C_0 $ of finite-state additive noise channels here is shown to be the exact value of $ C_{0f} $ for such channels in \cite{saberi2020explicit}. This paper goes well beyond \cite{saberi2020explicit}; in particular, it includes the estimation problem as well as erasure channels.

These bounds, in combination with the derived bounded state estimation conditions, yield separate necessary and sufficient conditions for achieving uniformly bounded state estimation errors via such channels, in terms of the  topological entropies of the linear system and the channel. Here, both continuous and discrete topological entropies appear in one equation, which can be intriguing for further studies. In a recent conference paper \cite{saberi2018estimation}, a worst-case {\em consecutive window} model of binary erasure channels was studied. Such a model can be made memoryless by a lifting argument, after which classical techniques can be applied. In contrast, the finite-state models studied here are state-dependent and require a different approach.

\subsection{Paper organization}
The rest of the paper is organized as follows. In section~\ref{sec:est}, the problem of linear state estimation over finite-state channels is considered, in the  nonstochastic framework of \cite{nair2013nonstochastic}. The zero-error capacity of such channels is characterized in terms of {\em nonstochastic information} (Theorem \ref{thm:fci}), leading to a necessary and sufficient condition for having uniformly bounded estimation errors (Theorem \ref{thm:finitem}). In Section~\ref{sec:fscm} general finite-state channels and their topological entropy are defined. Two classes of such channels are then described, namely finite-state erasure and additive noise channels. In Sections \ref{sec:est_ec} and \ref{sec:est_anc}, we return to the problem of state estimation over channels in these classes, and derive separate necessary and sufficient conditions for bounded estimation errors. 
Finally, concluding remarks and future extensions are discussed in section~\ref{sec:conclusions}.

\subsection{Notation}
Throughout the paper, $ q $ denotes the channel input alphabet size, logarithms are in base $q$, and coding rates and channel capacities are in symbols (or \emph{packets}) per channel use. The cardinality of a set is denoted by $ |\cdot| $. Define $ V_r^n(q):=1+{n \choose 1} (q-1) + \dots +{n \choose r}(q-1)^r=\sum_{i=0}^{r}{n \choose i}(q-1)^i$, where $ r \in \{0,1,\dots,n\} $. Let $ \mathbf{B}_l $ be an $ l $-ball $ \{ v: \norm{v} \leq l\} $ centred at origin with $ \norm{\cdot} $  denoting a norm on a finite-dimensional real vector space. The signal segment $ (x(t))_{t=a}^{t=b} $ is denoted by $ x(a:b) $. Further if there is no ambiguity in time segment, it is denoted by a vector $ \mathbf{x} $. The set of natural numbers including zero is denoted by $ \mathbb{N}_0 :=\mathbb {N}\cup \{0\} $.

\section{State Estimation over Finite-state Channels}\label{sec:est}
In this section, we first briefly provide some necessary background on the {\em uncertain variable} framework \cite{nair2013nonstochastic}. Next, the finite-state uncertain channel is defined and its zero-error capacity is characterized in terms of {\em maximin information}. Using these, a necessary and sufficient condition for linear state estimation with uniformly bounded estimation errors via finite-state uncertain channels is derived, extending the memoryless channel analysis in \cite{nair2013nonstochastic}. 

\subsection{Finite-state uncertain channels}
%Realizations of the inputs are defined in function space $ \mathscr{X} \subset \mathcal{X}^{\infty} $. 
Let $\Pi$ be a sample space. An {\em uncertain variable} $Z$ is a mapping from $ \Pi $ to a set  $\mathcal{Z}$. 
Given another uncertain variable $ W $,  the marginal, joint and conditional ranges are denoted by $
\llbracket Z\rrbracket:= \{ Z(\pi) : \pi \in \Pi \} \subseteq \mathcal{Z}, 
\llbracket Z,W\rrbracket:= \{ (Z(\pi),W(\pi)) : \pi \in \Pi \} \subseteq \llbracket Z\rrbracket\times \llbracket W\rrbracket,
\llbracket W|z\rrbracket:= \{ W(\pi) : Z(\pi)=z, \pi \in \Pi \},
\llbracket W|Z\rrbracket:= \{ \llbracket W|z\rrbracket : z \in \llbracket Z \rrbracket\}
$, respectively. 
Uncertain variable $Z$ and $W$ are said to be {\em mutually unrelated} if $\llbracket Z,W\rrbracket = \llbracket Z\rrbracket \times \llbracket W\rrbracket$, i.e., if the joint range is equal to the Cartesian product of the marginal ones.

In what follows, assume that $ \mathcal{X} $ and $ \mathcal{Y} $ are the input and output spaces of the channel, respectively. Now, a finite-state uncertain channel can be defined as follows.

\begin{defi} [Finite-state uncertain channel] \label{def:fuv}
	An uncertain channel with any admissible input sequence $x(0:t) \in \mathcal{X}^{t+1}$ and output sequence $y(0:t-1) \in \mathcal{Y}^{t}$ is said to be finite-state if for any $ t \in \mathbb{N}_0$ (setting $y(0:-1)=\emptyset$),
	\begin{align}
		\begin{split}
			\llbracket Y(t),S(t+1)|x(0:t),y(0:t-1),s&(0:t) \rrbracket=\\
			\llbracket &Y(t)|x(t),s(t) \rrbracket, 
		\end{split}\label{m1}
	\end{align}
	where $ s(t) \in \mathcal{S} $ is state of the channel at time $ t $ and $ \mathcal{S} $ is a finite set of states.
\end{defi}
In other words, given channel inputs, states and past outputs, the current output and the next state are {\em conditionally unrelated} with the past inputs, states and outputs. 

In a {\em zero-error code} $\mathcal{F}$, any two distinct codewords $ x(0:n),\, x'(0:n) \in \mathcal{F} $ can never result in the same channel output sequence, regardless of the channel noise and initial state. For a finite-state channel, the zero-error capacity is defined as follows.
\begin{defi} \label{def:c0}
	The zero-error capacity, $C_0$, is the largest block-coding rate that permits zero decoding errors, i.e., 
	\begin{align}
		C_0:=\sup_{n \in \mathbb{N}_0, \, \mathcal{F} \in {\mathscr{D}(n)}}
		\frac{\log |\mathcal{F}|}{n+1}, \label{c0def}
	\end{align}
	where $ {\mathscr{D}(n)}\subseteq \mathcal{X}^{n+1}$  is the set of all zero-error codes of length $ n+1 $. 
\end{defi}
Note that since the logarithms are in base $q$, where $q$ is the size of the input alphabet, $C_0$ takes maximum value one $q$-ary symbol per sample.
\subsection{Zero-error capacity and maximin information}
Consider the following definition.
\begin{defi}[Overlap Connectivity/Isolation] \textcolor{white}{.}
	\begin{itemize}
		\item A pair of points $ x $ and $ x' \in \llbracket X \rrbracket$ are {\em $\llbracket X|Y \rrbracket$-overlap connected} if a finite sequence of conditional ranges, $\{ \llbracket X|y_i \rrbracket \}_{i=1}^n$ exists such that $ x \in \llbracket X|y_1 \rrbracket $, $ x' \in \llbracket X|y_n \rrbracket $ and each conditional range has nonempty intersection with its predecessor, i.e., $ \llbracket X|y_i \rrbracket \cap \llbracket X|y_{i-1} \rrbracket \neq \emptyset $ for each $ i \in \{2,\dots,n\} $. Furthermore, a set $ \mathcal{B} \subset \llbracket X|Y \rrbracket $ is called $\llbracket X \rrbracket$-overlap connected if every pair of points in $ \mathcal{B} $ are overlap connected;
		\item A pair of sets $ \mathcal{B},\mathcal{C} \subset \llbracket X|Y \rrbracket $ are {\em $\llbracket X|Y \rrbracket$-overlap isolated} if there are no point in $ \mathcal{B}$ that is $\llbracket X|Y \rrbracket$-overlap connected with any point in $\mathcal{C}$;
		\item An {\em $\llbracket X|Y \rrbracket$-overlap isolated partition}, denoted by $\mathscr{P}$ (of $ \llbracket X \rrbracket $) is a  partition of $\llbracket X \rrbracket$ where every pair of distinct member-sets is $\llbracket X|Y \rrbracket$-overlap isolated.
		\item An {\em $\llbracket X|Y \rrbracket$-overlap partition} is an $\llbracket X|Y \rrbracket$-overlap isolated partition in which each member-set is $\llbracket X|Y \rrbracket$-overlap connected.
	\end{itemize}
\end{defi}
Furthermore, there exists a unique overlap partition $ \llbracket X|Y \rrbracket_* $ that satisfies $ |\mathscr{P}| \leq | \llbracket X|Y \rrbracket_*| $ for any $ \llbracket X|Y \rrbracket $, cf. \cite{nair2013nonstochastic} for a detailed treatment.

{\em Maximin information} is defined as 
\begin{align}
	I_*(X;Y):=\log \big| \llbracket X|Y \rrbracket_* \big|. \label{maximin}
\end{align}
Given an input sequence $ x(0:n) $ and the initial state $ s_0 $ which by Def. \ref{def:fuv} a finite-state uncertain channel maps to an uncertain output signal $ Y(0:n) $ so that $ \llbracket Y(0:n)|x(0:n),s_0 \rrbracket \in \mathcal{Y}^{n+1} $. Since the initial condition generally is not known to the encoder or decoder, it is considered as another source of uncertainty and thus
\begin{align}
	\llbracket X(0:n)|y(0:n) \rrbracket = \bigcup_{s_0 \in \mathcal{S}}\llbracket X(0:n)|y(0:n),s_0 \rrbracket. \label{condra}
\end{align}

We now can give the following theorem.
\begin{thm}\label{thm:fci} 
	For any finite-state uncertain channel~(Def. \ref{def:fuv}), 
	\begin{align}
		C_0= \sup_{n \in \mathbb{N}_0, X(0:n):\llbracket X(0:n) \rrbracket \in \mathcal{X}^{n+1} }\frac{1}{n+1} I_*[X(0:n);Y(0:n)]. \label{minmaxinfo}
	\end{align}
\end{thm}
\begin{proof}
	See Appendix \ref{app:cfi}.
\end{proof}
This result shows that the largest average rate that can be transmitted via a finite-state uncertain channel with zero decoding error coincides with the largest average maximin information rate across it. We note that this generalizes the corresponding result in  \cite{nair2013nonstochastic}, which was limited to channels without memory.
\subsection{State estimation of LTI systems over uncertain channels}
Consider a linear time-invariant (LTI) dynamical system 
\begin{align}
	\begin{split}
		X(t+1)&=AX(t)+V(t) \in\mathbb{R}^{n_x},\\
		Y(t)&=CX(t)+W(t) \in\mathbb{R}^{n_y},  \label{lti1}
	\end{split}
\end{align}
where $A$ and $C$ are constant matrices, and the uncertain variables $ V(t)$  and $ W(t)$ represent process and measurement disturbances. Here, the goal is to keep the estimation error \emph{uniformly bounded}. In this paper, this means that for any noise ranges $ \llbracket V(t) \rrbracket $ and $ \llbracket W(t) \rrbracket $ with $ \sup_{t \geq 0} \llbracket \norm{V(t)} \rrbracket < \infty$ and  $ \sup_{t \geq 0} \llbracket \norm{W(t)} \rrbracket < \infty$, $ \exists l >0 $ such that for any initial condition range $ \llbracket X(0) \rrbracket  \subseteq \mathbf{B}_l $, $ \sup \|\hat{X}(t)-X(t)\|$ is bounded. Here, $ \hat{X}(t)$ denotes the state estimate with $ \hat{X}(0)=0$  and the supremum is over all $t\geq 0$ and all valid noise and initial state realizations. The following assumptions are made:
\begin{itemize}
	\item[A1:] The pair $ (C,A) $ is observable;
	\item[A2:] There exist uniform bounds on the initial condition $X(0)$ and the noises $ V(t)$, $W(t)$;
	\item[A3:] The initial state $X(0)$, the noise signals $V$, $W$, and the channel error patterns are {\em mutually unrelated};
	\item[A4:]  The zero-noise sequence pair $ (V,W)=(0,0) $ is valid;
	\item[A5:] $A$ has one or more eigenvalues $\lambda_i$ with magnitude greater than one.
\end{itemize}

\tikzstyle{block1} = [draw, fill=celadon, rectangle, 
minimum height=3em, minimum width=5em]
\tikzstyle{block2} = [draw, fill=cblue, rectangle, 
minimum height=3em, minimum width=5em]
\tikzstyle{input} = [coordinate]
\tikzstyle{output} = [coordinate]
\tikzstyle{pinstyle} = [pin edge={to-,thin,black}]
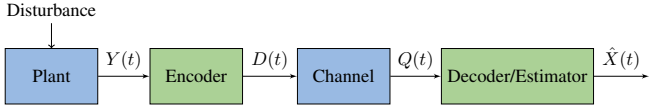
\begin{figure}[t]
	\centering
	\scalebox{.7}{\begin{tikzpicture}[auto, node distance=2cm,>=latex']
			\node [block2, pin={[pinstyle]above: Disturbance }] (plant) { Plant };
			\node [block1, right of=plant, node distance=2.75cm] (encoder) { Encoder };
			\node [block2, right of=encoder, node distance=2.8cm] (channel) {Channel };
			\node [block1, right of=channel, node distance=3.3cm] (decoder) { Decoder/Estimator} ;
			\node [output, right of =decoder, node distance=2.5cm] (output){};
			
			\draw [draw,->] (plant) -- node {$ Y(t) $} (encoder);
			\draw [->] (encoder) -- node {$ D(t) $} (channel);
			\draw [->] (channel) -- node {$ Q(t) $} (decoder);
			\draw [->] (decoder) -- node  {$ \hat{X}(t) $}(output);
	\end{tikzpicture}}
	\caption{\small{State estimation via a communication channel.}}
	\label{est}
\end{figure}	
Another way of formulating the system dynamics with bounded noise is as a {\em difference inclusion}; see, e.g. \cite{kellett2005robustness}.

The {\em topological entropy} of the linear system is given by 
\[
h_{lin}= \sum_{|\lambda_i|\geq 1} \log|\lambda_i|,
\]
and can be viewed as the rate at which it generates uncertainty. We have the following theorem.
\begin{thm} \label{thm:finitem}
	Consider an LTI system \eqref{lti1} satisfying conditions A1--A5. Assume that outputs are coded and estimated via a finite-state uncertain channel~(Def. \ref{def:fuv}) having zero-error capacity $ C_0 > 0$.  If there exists a coder-estimator that yields uniformly bounded estimation errors then
	\begin{align}
		C_0 \geq h_{lin}. \label{estdis}
	\end{align}
	Conversely, if \eqref{estdis} is satisfied as a strict inequality then there exists a coder-estimator that achieves uniformly bounded estimation errors.
\end{thm}
\begin{proof}
	See Appendix \ref{app:e}.
\end{proof}
\begin{rem}
	Theorem \ref{thm:finitem} extends the results of \cite{nair2013nonstochastic} for memoryless channels to finite-state channels.
	It states that uniformly reliable estimation is possible if and only if the zero-error capacity of the channel exceeds the rate at which the system generates uncertainty. The necessity argument of the proof relies on Theorem 1, which proves that any overlap isolated partition of $ X(0) $ cannot exceed the number of messages transmittable without error, and this does not depend on the channel state (memory) realization. Furthermore, sufficiency follows from showing the existence of a zero-error code with a rate arbitrarily close to  $ C_0 $. By choosing a large enough blocklength, the dependence of the code on the channel memory becomes negligible. 
\end{rem}
In the sequel, we consider finite-state channel models and give some preliminaries that we use later on.

\section{Finite-state Channel Models and Topological Entropy} \label{sec:fscm}
Consider a channel with output $y(t) \in \mathcal{Y} $ at time $ t $ which is a function of current input $ x(t) \in \mathcal{X}$ and correlated noise sequence $(v(t))_{t \in \mathbb{N}_0}$ governed by a finite-state machine or state transition graph. The directed graph describing the finite-state machine is defined as $\mathscr{G}=(\mathcal{S},\mathcal{E})$, where $\mathcal{S}=\{s_1,\dots,s_m\}$ denotes the vertex set ({\em states} of the channel) and $\mathcal{E}\subseteq\mathcal{S}\times \mathcal{S}$ denotes the edge set with $(s,s')\in\mathcal{E}$ capturing the possibility of state transition between $s,s'\in\mathcal{S}$. 

Let $ s_0 \in\mathcal{S}$ and $ x(0:n-1) $ denote the starting state and input sequence, respectively. Define the {\em adjacency matrix}  $ \mathcal{A} \in \{0,1\}^{|\mathcal{S}|\times |\mathcal{S}|}$ such that the $(s,s')$th entry $\mathcal{A}_{s,s'}$ equals 1 if the state of the channel can transition from $s$ to $s'$, and equals 0 otherwise. As an example, Fig. \ref{fig:mealy} shows a finite-state machine corresponding to a (noise) process $(v(t))_{t \in \mathbb{N}_0}$ at which no consecutive 1's can happen. Here,
\[\mathcal{A} =\begin{bmatrix}
	1 & 1\\
	1 & 0
\end{bmatrix}.\]
In symbolic dynamics, topological entropy  is defined as the asymptotic growth rate of the number of possible  state sequences. For a finite-state machine with an irreducible transition matrix\footnote{A matrix is irreducible if and only if its associated graph is strongly-connected \cite[Ch. 8]{brualdi2008combinatorial}.} $\mathcal{A}$, the topological entropy (in base $q$) is known to coincide with  $\log\lambda$ , where $\lambda$ is the {\em Perron} eigenvalue of $\mathcal{A}$~\cite{lind1995introduction}. This is essentially due to the fact that the number of the paths from state $ s_i $ to $ s_j $ in $n$ steps is the $ (i,j) $-th element of $ \mathcal{A}^n$, which grows proportionally with $\lambda^n$ for large $n$.

For a given initial state $s_0\in\mathcal{S}$, define the binary indicator vector $ \mathbf{z}_0\in\{0,1\}^{|\mathcal{S}|} $ consisting of all zeros except for a 1 in the position corresponding to initial state $s_0$; e.g. in Fig. \ref{fig:mealy}, if starting from state $ s_1 $, then $\mathbf{z}_0 =[1,0]^\top$. Let $ \mathscr{Y}(s_0,x(0:n-1)) $ denote the set of all output sequences that can occur by  transmitting the input sequence $ x(0:n-1) $ with initial channel state $ s_0 $. Observe that since each output of the channel (which can be a correctly received symbol or with error) triggers a different state transition, each sequence of state transitions has a one-to-one correspondence to the output sequence, given the input sequence. 

Based on these observations and Perron-Frobenius Theorem \cite[Thm. 4.2.3]{lind1995introduction}, we have the following result.

\begin{prop}[\cite{saberi2020explicit}]\label{thm:out}
	For a finite-state channel with irreducible adjacency matrix, there exist positive constants $ \alpha $ and  $ \beta $ such that, for any input sequence $ \mathbf{x}=x(0:n-1) \in \mathcal{X}^n $,
	\begin{align}
		\alpha \lambda^n \leq |\mathscr{Y}(s_0,\mathbf{x}) | \leq  \beta  \lambda^n , \label{outlam}
	\end{align}
	where $\lambda $ is the Perron eigenvalue of the adjacency matrix.
\end{prop}	
In other words, Proposition \ref{thm:out} shows that the evolution of output set, $ \mathscr{Y}(s_0,\mathbf{x}) $ size, starting from any initial state, is controlled by the maximum eigenvalue of the adjacency matrix. We denote the topological entropy of the channel (which is logarithm of the maximum eigenvalue) with $ h_{ch}:=\log \lambda $.

We study channels with two types of errors; erasure and additive noise which are generalizations of binary erasure and symmetric channels, respectively. Based on the type of error, we separate channel models into erasure and additive noise types which are defined as follows.
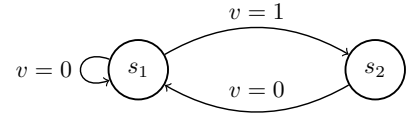
\begin{figure}[t]
	\centering
	\scalebox{0.9}{\begin{tikzpicture}
			[->, auto, semithick, node distance=3.5cm]
			\tikzstyle{every state}=[fill=white,draw=black,thick,text=black,scale=1]
			\node[state]   (S1)       {$ s_1 $};
			\node[state]   (S2)[right of=S1]   {$ s_2 $};
			%\node[state]   (S3)[right of=S2]   {$ s(t)=s_3 $};
			\path
			(S1) edge[out=160, in=200,looseness=5] node [left] {$ v=0 $}(S1)
			edge[bend left] node [above] {$ v=1 $}  (S2)
			(S2) edge[bend left]  node [above=0.1] {$ v=0 $}  (S1);
	\end{tikzpicture}}
	\caption{\small{A finite-state machine describing the transition of a noise process in a channel at which no consecutive 1's can happen.}}
	\label{fig:mealy}
\end{figure}
\begin{defi}[Finite-state erasure channels] \label{def:fsechannel}
	A channel is called {\em finite-state erasure} if the output $y(t) \in \mathcal{Y} := \mathcal{X} \cup \{*\}$ at time $t \in \mathbb{N}_0$ is obtained by
	\begin{align*}
		y(t) = \phi(x(t),v(t)):= \begin{cases}
			x(t), & \text{ if } v(t)=0, \\
			*, & \text{ otherwise},
		\end{cases},
	\end{align*}\\		
	where $*$ denotes erasure in the channel and the noise $ v(t) \in \{0,1\}$ is governed by a finite-state machine such that each outgoing edge from a state $ s(t) $ corresponds to different values $ v(t) $ of the noise starting. Thus, there are at most two outgoing edges from each state. Moreover, it is assumed the graph is strongly connected. 
\end{defi}
\begin{rem}
	The strong connectivity of the underlying graph is needed to be able to use the Perron-Frobenius Theorem. In a stochastic setting, strong connectivity is a common assumption leading to recurrence, see, e.g. \cite{sabag2016single}.
\end{rem}

The above finite-state machine is similar to the Mealy machine, where each edge is denoted by input and output pair (see, e.g. \cite[Ch. 2 ]{holcombe1982algebraic}). However, in this model, the input (to the communication channel) has no effect on the transitions, so it is ignored in the representation of Fig. \ref{fig:mealy}, and the noise (as the output of the finite-state machine) is shown only.\footnote{Note that each state may have different possible noise (output) set, e.g. in Fig. \ref{fig:mealy}, there is one outgoing edge from state $ s_2 $ labelled with $ v=0 $ which means no error occurs visiting this state. Also, $ s_1 $ has two outgoing edges, which mean visiting this state may cause an error or error-free transmission over the channel.} 

In finite-state erasure channels, the erasure appears as an extra symbol in the output. Therefore, the receiver knows the locations of the erased symbols; however, this information is not available to the sender.

\begin{defi}[Finite-state additive noise channels]
	\label{def:fanchannel}	
	A channel is called {\em finite-state additive noise} if the output at time $t$ is obtained by
	\begin{align}
		y(t) = x(t) \oplus v(t) \in \mathcal{X}, \, t \in \mathbb{N}_0, \label{fan}
	\end{align}
	where the correlated additive noise $ v(t) $  is governed by a finite-state machine such that each outgoing edge from a state $ s(t) $ corresponds to different values $ v(t) $ of the noise. Thus, there are at most $ |\mathcal{X}| $ outgoing edges from each state.
	We assume the state transition diagram of the channel is strongly connected.
\end{defi}
The finite-state channels in Defs.~\ref{def:fsechannel} and~\ref{def:fanchannel} generalize their stochastic counterparts, the binary erasure and symmetric channels, respectively. Here, instead of having a probability of error for every single use of the channel (memoryless channel), the errors may occur based on a finite-state machine. 

The state process described by the finite state machine is an uncertain Markov chain \cite{nair2013nonstochastic}, which, in a stochastic setting, corresponds to a topological Markov chain \cite[Ch.2]{renyi1970foundations}. This is a more general property than Markovianity. It allows the next state probability to depend not just on the current state, but also on previous states. Our results remain valid in these situations, since the zero-error capacity is not a function of the transition probabilities, but only of the topological structure of the state machine.

In the remainder of this section, some examples of finite-state channels with no stochastic assumptions are given which are studied in the recent literature. Such models usually arise in worst-case approaches where no probability assumptions are made, and some bounds on errors are considered \cite{badr2017layered,wang2018end}. Here, an error can refer to both erasure (in erasure channels) or error (in additive noise channels). We show that these channels can be modelled as a finite-state machine. %In the sequel, we discuss some of the examples for the channels that can be modelled this way.
\begin{figure}[t]
	\includegraphics[width=60mm]{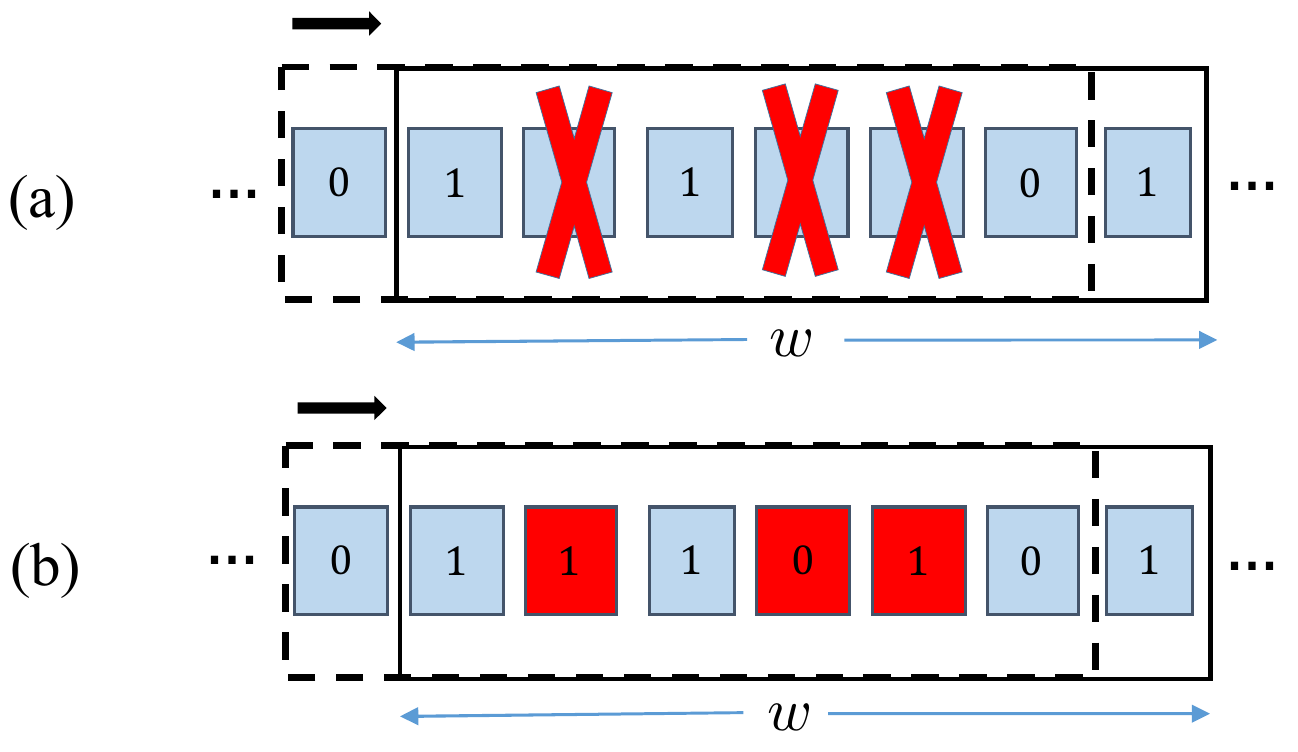}
	\centering
	\caption{\small{Bounded error structure for binary $(w=7,d=3)$ sliding-window erasure~(a) and sliding-window symmetric~(b) channels. }}
	\label{fig:slidep}
\end{figure}	
\subsection{Sliding-window channel models} \label{swcm}
A channel is called $(w,d)$ {\em sliding-window} if the number of errors that may occur over a sliding-window of length $w$ is upper bounded by a non-negative integer $ d $. The maximum error rate is then $d/w$. Two channel models for erasure and additive noise cases are considered and are referred to as {\em sliding-window erasure} and {\em sliding-window symmetric} channels. In the symmetric channel, the input symbol may get mapped to any symbol in the output alphabet and an error occurs when the received symbol is different than the input one.\footnote{The symmetric structure can be seen as a special case of the additive noise channel model where the symbols may get mapped to some subsets of the output alphabet.}  Figure~\ref{fig:slidep} illustrates simple sliding-window erasure and symmetric channels with binary input alphabets, for the case of $ w=7$ and $ d =3$. 
The channel output depends on the errors in the previous window; thus these channels have memory.
Equivalently, they can be represented as finite-state channels.

For a sliding-window erasure channel, the current state of the channel is naturally represented as a $w$-bit word, with $*$ and $\circ$, respectively, indicating the locations of erroneous and error-free transmissions in the previous window. 
Note that this state depends only on the locations of the errors in the past $w$ symbols, not on the transmitted symbols. Moreover, all the possible binary locations of erasures in a sequence of length $w$ gives the number of states which is selecting $i=0, 1,\ldots , d$ erasure locations, i.e., $|\mathcal{S}|=V_d^w(2)=1+{w \choose 1} + \dots +{w \choose d}$.

Due to restrictions on the number of errors in each sliding window, not all states can be visited from any starting state in a single step. For example,  Fig.~\ref{fig:transition} shows the state transition diagram for the $(w=3,d=1)$ sliding-window erasure channel. Here, instead of labelling the edges, red-dashed edges correspond to erasure ($ v(t)=1 $) and solid edges are for error-free ($ v(t)=0 $) transmission over the channel.
\begin{figure}[t]
	\centering
	\begin{tikzpicture}[->, >=stealth', auto, semithick, node distance=1.7cm]
		\tikzstyle{every state}=[fill=white,draw=black,thick,text=black,scale=0.9]
		\node[state]   (S1)       {$  s_1 $};
		\node[state]    (S2)[left of=S1]   {$ s_2 $};
		\node[state]    (S3)[below of=S2]   {$ s_3 $};
		\node[state]    (S4)[right of=S3]   {$ s_4 $};
		\path
		(S1) edge[out=30,in=70,looseness=8] (S1)
		%		(S1) edge[loop right]    (S1)
		edge[dashed, red,bend right]   (S2)
		(S2) edge[bend right]   (S3)
		(S3) edge[bend right]   (S4)
		(S4) edge[bend right]   (S1)
		(S4) edge[dashed, red]   (S2);
	\end{tikzpicture}
	\scalebox{0.8}{\begin{tabular}{c c} 
			\hline
			State & Binary representation \\  
			$ s_1 $ & $ \circ \circ \circ $  \\ 
			
			$ s_2 $ & $ \circ \circ  * $  \\
			
			$ s_3 $ & $ \circ  * \circ $  \\
			
			$ s_4 $ & $ * \circ \circ $   \\ \hline 
	\end{tabular}}
	\caption{\small{States transition graph governing the noise in a $ (w =3,d = 1) $ sliding-window erasure channel.}}
	\label{fig:transition}
\end{figure}
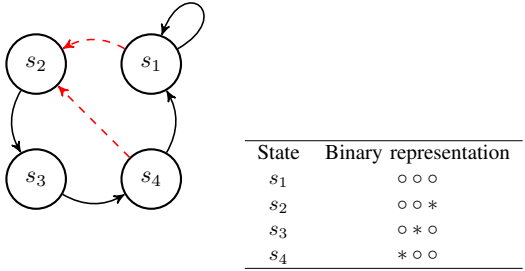

For an $(w,d)$ sliding-window symmetric channel, we define the state as a $q$-ary word of length $w$, in which $\circ$ indicates no error and $\circ',\circ'',\ldots$ label
the $(q-1)$ erroneous symbol swaps that can occur. %\footnote{Equivalently, by writing the channel input-output relationship as \eqref{fan}, where $x(t), y(t), v(t) \in\mathbb{Z}_q$, the current channel state is equivalent to $v(t-w:~t-~1)\in\mathbb{Z}_q^w$, with at most $d$ nonzero entries.}
This is not the most compact state representation; however, for a given input sequence it  yields a one-to-one relationship between the state and  output sequences, which will be useful in deriving lower bounds on zero-error capacity. The set $\mathcal{S}$ of possible states can be shown to be of size $V^w_d(q)$
(by selecting $i=0, 1,\ldots , d$ error locations, each with $q-1$ distinct possibilities, in a window of length $w$).

To illustrate this, see Fig.~\ref{fig:trans2} for the possible states and transitions for a $(w=3,d=1)$ sliding-window symmetric channel with alphabet size $q=3$.
For example, the state $s_1$ is error-free and can have $q=3$ transitions: (\textit{i}) there is no error, resulting in no change in the state of the channel, (\textit{ii}) there is an error with output $(x(i)+1)\mod 3$, resulting in transition to state $s_2$, and (\textit{iii}) there is an error with output $(x(i)+2)\mod 3$, resulting in transition to state $s_5$. Note that both $s_2$ and $s_5$ represent only one error; hence, in the case of the sliding-window erasure channel, they could be combined into one state.
\subsection{Bursty errors}
In this channel, errors happen up to $d$ consecutive symbols in a sliding window of length $w$. An example of this channel which has up to 2 errors in a sliding window of size 3 is shown in Fig. \ref{fig:models} (a). The states of this channel can be defined similarly to Section \ref{swcm}.
\subsection{Guard space between errors}
In this channel, after each bounded number of errors, there is a minimum number of error-free transmissions. Therefore, a guard space between errors is present. Figure \ref{fig:models} (b) shows an example of this channel in which between any errors (of maximum length 2) a guard space of length 3 exists.
\begin{figure}[t]
	\centering
	\begin{tikzpicture}[->, >=stealth', auto, semithick, node distance=2cm]
		\tikzstyle{every state}=[fill=white,draw=black,thick,text=black,scale=.75]
		\node[state]   (S1)       {$  s_1 $};
		\node[state]    (S2)[below left of=S1]   {$ s_2 $};
		\node[state]    (S3)[left of=S1]   {$ s_3 $};
		\node[state]    (S4)[below of=S2]   {$ s_4 $};
		\node[state]    (S5)[below right of=S1]   {$ s_5 $};
		\node[state]    (S6)[right of=S1]   {$ s_6 $};
		\node[state]    (S7)[below of=S5]   {$ s_7 $};
		\path
		(S1) edge[out=70,in=110,looseness=8] (S1)
		%		(S1) edge[loop right]    (S1)
		edge[bend right,dashed, red]  node [above] {$ 1 $} (S2)
		edge[bend left,dashed, red]  node [above] {$ 2 $}  (S5)
		(S2) edge[]   (S3)
		(S3) edge[bend right]   (S4)
		(S4) edge[]   (S1)
		(S4) edge[dashed, red] node [left] {$ 1 $}  (S2)
		(S4) edge[dashed, red]  node [below left] {$ 2 \quad \; $} (S5)
		(S7) edge[dashed, red]  node [below right]  {$ \quad \; 1 $} (S2)
		(S5) edge[]   (S6)
		(S6) edge[bend left]   (S7)
		(S7) edge[]   (S1)
		(S7) edge[dashed, red]  node [right] {$ 2 $} (S5);
	\end{tikzpicture}\qquad
	\scalebox{.85}{\begin{tabular}{c c}
			State & Representation\\  
			\hline
			$ s_1 $ & $ \circ \circ \circ $  \\ 
			$ s_2 $ & $ \circ \circ  \circ' $  \\
			$ s_3 $ & $ \circ \circ'\circ $  \\
			$ s_4 $ & $ \circ'\circ \circ $   \\
			$ s_5 $ & $ \circ \circ \circ'' $  \\
			$ s_6 $ & $ \circ\circ'' \circ $  \\
			$ s_7 $ & $\circ'' \circ \circ $   \\\hline 
	\end{tabular}}
	\caption{\small{States and transition graph governing the noise in a $(3,1)$ sliding-window symmetric channel with $ q= 3$. Here $ \circ' $ and $ \circ'' $ denote possible swaps with respect to sent symbol. Each labeled transition edge corresponds to the non-zero noise ($ v(t)\neq 0 $).}}
	\label{fig:trans2}
\end{figure}
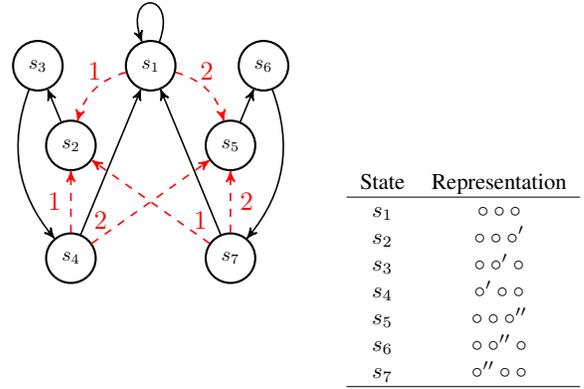
\subsection{Gilbert-Elliot like channels}
The celebrated Gilbert-Elliot channel is used to describe bursty errors and in its stochastic form is modelled by two states in transmission. In the {\em good} state, the probability of error is low, but in the {\em bad} state the probability of error is high. In other terms, the Gilbert-Elliot channel has two main operation states, wherein the good state errors in the channel are sparse, while in the bad state, it behaves as a burst-error channel. Recently, in \cite{badr2017layered}, Badr \emph{et. al.} have introduced a worst-case model which is an approximation of stochastic Gilbert-Elliot channels. In their proposed model in any sliding window of length $ W $, the channel can have two patterns: (i) a single burst or consecutive erasure of maximum length $ B $, or (ii) a maximum of $ N $ erasures in arbitrary positions within the window. They use the notation $ C(N, B, W ) $ for such a channel. Again, this model can be described with the introduced finite-state machine. Figure \ref{fig:models} (c) shows the state transition of a $ C(1, 3, 4 ) $ channel.		
\section{State Estimation over Finite-state Erasure Channels} \label{sec:est_ec}
In this section, we investigate bounded state estimation of an LTI system over the finite-state erasure channels introduced in the previous section. It is shown that the topological properties of the channel state transition diagram are linked with the topological entropy of the LTI system in order to estimate the states with bounded error. Note that all topological entropies and capacities are measured in base $ q $, i.e. in units of $q$-ary symbols per sample.

We give the following theorem.
% Note that $ C_0 $ is generally hard to compute \cite{korner1998zero}. Therefore, in the sequel, we drive separate necessary and sufficient conditions for bounded state estimation.
%		Further, note that the finite-state erasure and additive noise channels considered in this paper are finite-state uncertain channels (Def. \ref{def:fuv}).

\begin{prop} %[Bounded estimation errors via finite-state erasure channel]
	\label{prop:nseest} 
	Consider an LTI system in \eqref{lti1} satisfying conditions A1--A5. Assume that the measurements are coded and transmitted via a finite-state erasure channel~(Def. \ref{def:fsechannel}) with topological entropy $h_{ch}$ and maximal ratio $ \tau $. Then uniformly bounded estimation errors can be achieved if
	\begin{align}
	h_{lin} + h_{ch}+\tau<1 \label{stabi3}.
	\end{align}
	Conversely, there exist sequences of process and measurement noise for which the estimation error grows unbounded if
	\begin{align}
	h_{lin}+\tau>1. 
	\end{align}
\end{prop}
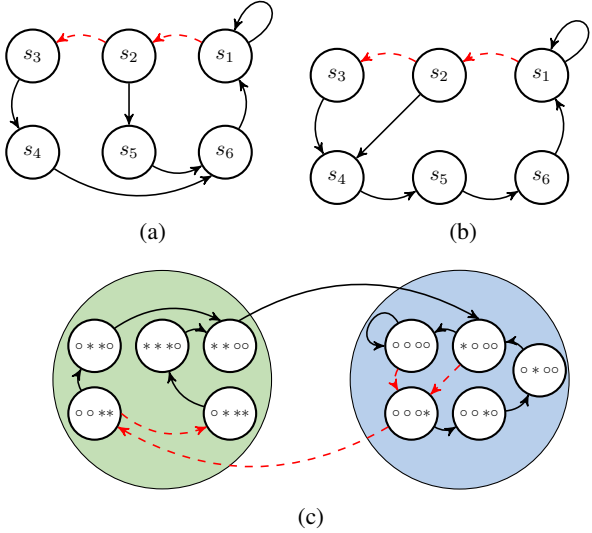
\begin{figure}[t]
	\centering 
	\subcaptionbox{}{\begin{tikzpicture}[->, >=stealth', auto, semithick, node distance=1.6cm]
		\tikzstyle{every state}=[fill=white,draw=black,thick,text=black,scale=0.8]
		\node[state]   (S1)       {$  s_1 $};
		\node[state]    (S2)[left of=S1]   {$ s_2 $};
		\node[state]    (S3)[left of=S2]   {$ s_3 $};
		\node[state]    (S4)[below of=S3]   {$ s_4 $};
		\node[state]    (S5)[right of=S4]   {$ s_5 $};
		\node[state]    (S6)[right of=S5]   {$ s_6 $};
		\path
		(S1) edge[out=30,in=70,looseness=8] (S1)
		edge[dashed, red,bend right]   (S2)
		(S2) edge[dashed, red,bend right]   (S3)
		edge[]   (S5)
		(S3) edge[bend right]   (S4)
		(S4) edge[bend right]   (S6.240)
		(S5) edge[bend right]   (S6)
		(S6) edge[bend right]   (S1);
		\end{tikzpicture}}
	\subcaptionbox{}{\begin{tikzpicture}[->, >=stealth', auto, semithick, node distance=1.7cm]
		\tikzstyle{every state}=[fill=white,draw=black,thick,text=black,scale=0.8]
		\node[state]   (S1)       {$  s_1 $};
		\node[state]    (S2)[left of=S1]   {$ s_2 $};
		\node[state]    (S3)[left of=S2]   {$ s_3 $};
		\node[state]    (S4)[below of=S3]   {$ s_4 $};
		\node[state]    (S5)[right of=S4]   {$ s_5 $};
		\node[state]    (S6)[right of=S5]   {$ s_6 $};
		\path
		(S1) edge[out=30,in=70,looseness=8] (S1)
		edge[dashed, red,bend right]   (S2)
		(S2) edge[dashed, red,bend right]   (S3)
		edge[]   (S4)
		(S3) edge[bend right]   (S4)
		(S4) edge[bend right]   (S5)
		(S5) edge[bend right]   (S6)
		(S6) edge[bend right]   (S1);
		\end{tikzpicture}} \\ \vspace{3mm}
	\subcaptionbox{}{\begin{tikzpicture}[->, >=stealth', auto, semithick, node distance=1.5cm]
		\tikzstyle{every state}=[fill=white,draw=black,thick,text=black,scale=.6]
		\tikzstyle{surround} = [circle,fill=cblue!70,draw=black]
		\node[state]    (S1)[]      {$  \circ\circ\circ\circ $};
		\node[state]    (S2)[below of=S1]   {$ \circ\circ\circ* $};
		\node[state]    (S3)[right of=S2]   {$ \circ\circ*\circ $};
		\node[state]    (S4)[above right=0.07cm and .3cm of S3]   {$ \circ*\circ\circ $};
		\node[state]    (S5)[right of=S1]   {$ *\circ\circ\circ $};
		\node[state]    (S7)[left=1.7cm of S1]   {$ **\circ\circ $};
		\node[state]    (S8)[left of=S7]   {$ ***\circ $};
		\node[state]    (S9)[left of=S8]   {$ \circ**\circ$};
		\node[state]    (S10)[below of=S9]   {$ \circ\circ** $};
		\node[state]    (S11)[below of=S7]   {$ \circ*** $};
		\path
		(S1) edge[out=120,in=180,looseness=3.5] (S1)
		edge[dashed, red,bend right]   (S2)
		(S2) edge[dashed, red,bend left]   (S10)
		edge[bend right]   (S3)
		(S3) edge[bend right]   (S4)
		(S4) edge[bend right]   (S5)
		(S5) edge[bend right]   (S1)
		edge[dashed, red]   (S2)
		(S7.70) edge[bend left]   (S5.90)
		(S8) edge[bend left]   (S7)
		(S9) edge[bend left]   (S7.110)
		(S10) edge[bend left]   (S9)
		(S10.0) edge[dashed, red,bend right]   (S11)
		(S11) edge[bend left]   (S8);
		\begin{pgfonlayer}{background}
		\node[surround,xshift=-6pt] (good)[fit =(S1)(S2)(S3)(S4)(S5),scale=0.9] {};
		\node[surround] (bad)[fit =(S7)(S8)(S9)(S10)(S11),fill=celadon!70,scale=0.88] {};
		\end{pgfonlayer}
		\end{tikzpicture}}
	\caption{\small{States transition diagram {governing} the noise sequence in (a) a channel that may have bursty errors up to length two in each sliding window of size 3, (b) a channel that has a guard space between any error with maximum length two, and (c)  Gilbert-Elliot like channel. {The red-dashed and black lines corresponds to transmissions with and without error,  respectively. B}lue circled group of states is \textit{good} group of states and green circled one is \textit{bad} group of states.}}
	\label{fig:models}
\end{figure}	
\begin{rem}
	The achievability part of this proposition involves the topological entropies of both the linear system and the channel. If their sum, which can be regarded as a total rate of uncertainty generation, is less than the worst-case rate at which symbols can be transported without error across the channel, i.e., $1-\tau$, then uniformly bounded estimation errors are possible. This can be seen as a {\em small-uncertainty} version of the small-gain theorem.
\end{rem}
To prove {Proposition \ref{prop:nseest}, we derive bounds on $ C_0 $ of the finite-state erasure channel.	
	Explicit formulas for the zero-error capacity typically do not exist except in special cases, even for memoryless channels. We derive an upper bound on $C_0$ by calculating the zero-error feedback capacity and a lower bound by constructing a zero-error code.}

In the following subsection, the zero-error feedback capacity of the finite-state erasure channel is investigated.
\subsection{Zero-error feedback capacity of finite-state erasure channel}

The zero-error feedback capacity $C_{0f}$ is the zero-error capacity of a channel in the presence of a noiseless feedback from the output.

Let $ \mathcal{M} $ be the message set, $ x(t) \in \mathcal{X} $ be the channel input and $ y(t) \in \mathcal{Y} $ the channel output. For  channels with feedback, the input to the channel is a function of the message and previous outputs, i.e., $ x(t)=f_t(m, y(0:t-1)), \, t=0,1,\dots, n,\,\, m\in \mathcal{M}$, where $ f_t(\cdot) $ is the encoding function (setting $y(0:-1)=\emptyset$). The set of functions $ \mathcal{F}=\{f_{0:n}(m,\cdot)|m\in \mathcal{M}\} $ is called a {\em zero-error feedback code} if no two distinct messages $ m \neq m'  \in \mathcal{M} $ can result in the same channel output sequence, regardless of the channel noise and initial state.
\begin{defi} \label{def:c0f}
	The zero-error feedback capacity of a channel is 
	\begin{align*}
	C_{0f}:=\sup_{n \in \mathbb{N}_0, \,  \mathcal{F} \in  {\mathscr{F}(n)}}
	\frac{\log |\mathcal{M}|}{n+1},
	\end{align*}
	where $\mathcal{F}$ is chosen from the set {$ \mathscr{F}(n)$ of all zero-error feedback codes with blocklength $ n+1 $.}
\end{defi}
Unlike the zero-error capacity, in zero-error feedback capacity, the encoder has access to the previous channel outputs. Since the family of zero-error codes with feedback includes the family of zero-error codes  without feedback, we have $C_0 \leq C_{0f}$. 

For discrete memoryless channels, $C_{0f}$ can be obtained through an optimization problem~\cite{shannon1956zero}. For finite-state channels with causal state information at the transmitter and receiver, by adopting Shannon's approach, it has been shown that $C_{0f}$ can be obtained by solving a dynamic programming problem~\cite{zhao2010zero}.	
In the finite-state erasure channel, the state is revealed by the output sequence. Thus, in the presence of an error-free feedback channel, the state is known to the encoder and the decoder, and we can apply the techniques of \cite{zhao2010zero}.

A {\em walk} in the state diagram is a contiguous sequence of directed edges, and corresponds one-to-one with a valid noise sequence. Each channel noise sequence starting from state $ s_0 $, i.e., $ \mathbf{v}:=v(0:n-1) \in \llbracket V(0:n-1)|s_0 \rrbracket$,  is equivalent to a walk through the finite-state machine graph denoted by a tuple (or {\em list}) $ \varpi:=(\xi(i))_{i=0}^{n-1} $ where $ \xi(i) \in \mathcal{E}$ is an edge between two vertices or states. Moreover, $ E(\varpi)$ or $E(\mathbf{v}) $ denotes the number of erasure edges ($ v(i)=1 $) in this walk.

Before giving the main result of this section, we give the following definition.
\begin{defi} [Maximal ratio] \label{def:mr}
	For any cycle in the state diagram of a finite-state channel and let $ \tau_i:=e_i/l_i $ where $ e_i $ is the number of error edges and $ l_i $ is the total number of edges in the cycle. The \emph{maximal ratio} is defined as the maximum $ \tau_i $ over all cycles, i.e., $\tau := \max_i \{\tau_i\}$.
\end{defi}
In the following lemmas, we show that when $n$ is large, walks on the cycle with maximal ratio lead to the maximum number of erasures.
\begin{lem} \label{lem:maxer}
	For the finite-state erasure channel with maximal ratio $ \tau=e/l $ and initial state $ s_0 $, the number of erasures $ E(\mathbf{v}) $ in any sequence $ \mathbf{v} $ is upper bounded by
	\begin{align}
		E(\mathbf{v}) \leq \tau n+ |\mathcal{S}|. \label{maxeras}
	\end{align}
\end{lem}
\begin{proof} See Appendix \ref{app:a}.
\end{proof}
Now we show that there are some sequences that get close to the upper bound in \eqref{maxeras} when $n$ is large enough.
\begin{lem} \label{lem:miner}
	For the finite-state erasure channel with maximal ratio $ \tau=e/l $ and initial state $ s_0 $, there is a sequence $ \mathbf{v} $ such that
	\begin{align}
		E(\mathbf{v}) > \tau n -l- |\mathcal{S}|. \label{worsteras}
	\end{align}
\end{lem}
\begin{proof} See Appendix \ref{app:b}.
\end{proof}
According to the above lemmas, for a finite-state erasure channel, we have the following 
\begin{align}
	\lim_{n \rightarrow \infty} \max_{s_0,\mathbf{y}} \frac{E(\mathbf{v})}{n}  =  \tau. \label{ratio}
\end{align}

The above lemmas are used to give the following formula for finite-state erasure channels.
\begin{thm} \label{thm:nsec0f} The zero-error feedback capacity of a finite-state erasure channel~(Def. \ref{def:fsechannel}) is 
	\begin{align}
		C_{0f} =1-\tau.\label{nsec0f}
	\end{align}
\end{thm}
\begin{proof}
	See Appendix \ref{app:c}.
\end{proof}
Theorem~\ref{thm:nsec0f} states that the zero-error capacity of the finite-state erasure channel with feedback coincides with the minimum fraction of the $q$-ary packets that may be successfully received. 

We now relate the zero-error capacity of the channel to its topological entropy.

\subsection{Zero-error capacity bounds for finite-state erasure channel}
In the following theorem, we show a lower bound linking the topological entropy of the channel to the zero-error capacity. The upper bound is from Theorem \ref{thm:nsec0f}.

\begin{thm}\label{thm:ec}
	The zero-error capacity of a finite-state erasure channel~(Def. \ref{def:fsechannel}) with topological entropy $h_{ch}$ and maximal ratio $ \tau $~(Def. \ref{def:mr}) is bounded by
	\begin{align}
	1-\tau-h_{ch} \leq C_0 \leq 1-\tau. \label{ecc0}
	\end{align}
\end{thm}
\begin{proof}
	See Appendix \ref{app:nselw}.
\end{proof}

\begin{rem}
	The topological entropy $h_{ch}$ can be viewed as the rate at which the channel dynamics generate uncertainty. Intuitively, this uncertainty cannot increase the zero-error capacity of the channel, which explains why it appears as a negative term in \eqref{ecc0}.
\end{rem}

\begin{rem}
	There are various results that bound $ h_{ch} = \log\lambda$. For instance, for any graph with maximum out-degree $ d_{max} $ and minimum out-degree $ d_{min} $, we have $d_{min} \leq \lambda \leq d_{max}$~\cite{berman1994nonnegative}. Therefore, a loose lower bound would be $1-\tau-\log d_{max}$. Moreover, note that $ d_{max}=2 $ for the state diagram of any finite-state erasure channel. Thus for large  alphabet size $q$, the lower bound meets the upper bound obtained in \eqref{nsec0f}, i.e., $\lim_{q \to \infty}C_0= 1-\tau$. In other words, for large packet sizes, the zero-error capacity is equal to the fraction of the packets that are not dropped in the worst-case scenario. This is comparable with the (small-error) capacity of the erasure channel which is equal to the expected fraction of the packets that are not dropped.
\end{rem}
\begin{rem}
	The proof of {Proposition} \ref{prop:nseest} follows from Theorems \ref{thm:finitem} and \ref{thm:ec}.
\end{rem}
In what follows, the zero-error capacity of the sliding-window erasure channel as an example of finite-state erasure channel is discussed and compared with the general results of Theorem \ref{thm:ec}.
\subsection{Example: sliding-window erasure channel}
Based on the sliding-window erasure channel's structure some new bounds are derived to compare with the general bounds discussed in the previous section.

Consider the structure of the state diagram of a $ (w,d) $ sliding window erasure channel. The cycles with maximal ratio~(Def. \ref{def:mr}) are the ones corresponding to the maximum number of erasures in the past $ n $ transmission. The following Lemma gives the reason.
\begin{lem} \label{lem:tec}
	For a $ (w,d) $ sliding-window-erasure channel $ \tau~=~d/w $.
\end{lem}
\begin{proof}
	See Appendix \ref{app:tec}.
\end{proof}
Considering Lemma \ref{lem:tec} and bounds in \eqref{ecc0} yields the following bounds for the sliding-window erasure channel.
\begin{cor}\label{cor:swtb}
	{For a $(w,d)$ sliding-window erasure channel with topological entropy $h_{ch}$, the following holds.}
	\begin{align}
	1-\frac{d}{w}-h_{ch} \leq C_0 \leq 1-\frac{d}{w}. \label{swec0}
	\end{align}
\end{cor}
\begin{rem}
	According to Theorem \ref{thm:nsec0f}, $ C_{0f}=1-d/w $  for sliding-window erasure channel. However, in this case, it is straightforward to see that the zero-error feedback capacity is upper bounded by $1-d/w$. This is because, for long input sequences, in a worst-case scenario $1-d/w$ proportion of symbols can be erased which  bounds the rate from above. Furthermore, this rate can be achieved by a simple feedback encoding method that re-transmits every erased symbol until it is successfully received. Therefore, the zero-error feedback capacity equals $ 1-d/w $.
\end{rem}
\begin{rem}
	In \cite{fong2019optimal}, the authors use the similar structure of the sliding-window erasure channel as here and use \textit{maximum distance separable}~(MDS) codes to achieve the rate of $ 1-d/w $ without feedback. A subtle but critical point to note here is that only for a few combinations of $(w,d,q)$ such codes exist; e.g. there is no MDS code for binary alphabets and roughly these codes exist for large alphabet sizes \cite{huffman2010fundamentals}. 
\end{rem}
Next, we derive a different lower bound for the zero-error capacity of the channel. 
\begin{thm}\label{thm:c0lec} The zero-error capacity of a sliding  window-erasure channel which arbitrarily erases up to $ d $ symbols in every sliding window of $ w $ symbols is lower bounded by 
	\begin{align}
	C_0 &\geq 1-\frac{1}{w}\log V^w_{d}(q) .\label{bounds_sw}
	\end{align}
\end{thm}
\begin{table*}[t]
	\centering
	\caption{\small{Bounds on the zero-error capacity of the studied channels.}}
	\scalebox{1}{\begin{tabular}{ c | c | c | c} 
			\hline
			Channel & Type & Lower bound & Upper bound \\ \hline  
			\multirow{2}{*}{Finite-state erasure} &  General (Thm. \ref{thm:ec}) & $ 1-\tau-h_{ch} $ & $ 1-\tau$  \\ 
			&Sliding-window erasure  (Thm. \ref{thm:c0lec})& $1-(1/w)\log V^w_{d}(q) $ & $ 1-d/w $  \\ \hline
			
			\multirow{2}{*}{Finite-state additive noise} &  General  (Thm. \ref{thm:fanctp}) &  $ 1-2h_{ch} $ & $ 1-h_{ch} $  \\
			
			&Sliding-window symmetric  (Thm. \ref{thm:nssc0f}) & $1-(1/w)\log V^w_{2d}(q)  $ & $ 1-(d/w) \log(q-1) $   \\ \hline 
	\end{tabular}}
\end{table*} 
\begin{figure}[t]
	\centering
	\includegraphics[width=75mm]{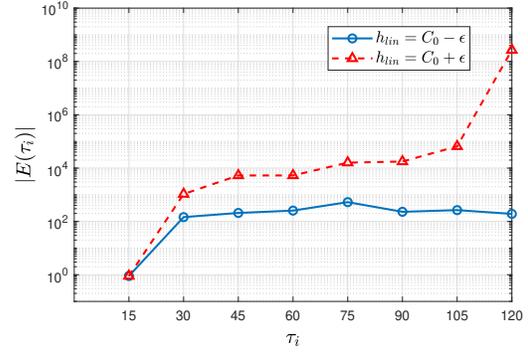}
	\caption{\small{State estimation errors over a $ (3,1) $ sliding-window erasure channel.}}
	\label{fig:NE}
\end{figure}
		\begin{proof}
			See Appendix \ref{app:swe}.
		\end{proof}
		Figure \ref{fig:cap_fse}(a) demonstrates the capacity bounds discussed so far, for a binary ($ q=2 $) sliding-window erasure channel. Note that the general topological entropy lower bound gives a tighter lower bound for small $ d/w $.	
\subsection{Numerical example}
Consider a scalar plant 
$
X(t+1)=a X(t)+V(t), |V(t)|<1, |X(0)|<1. %\label{explant}
$
Assume that the state can be measured with no noise, i.e., $ Y=X $. The state is estimated over a binary $ (3,1) $ sliding window erasure channel (see Fig. \ref{fig:transition}) for which $ \tau=1/3 $. For this channel $ C_0=2/3 $. The upper bound $ C_0 \leq 1 - \tau=2/3 $ is clear from Theorem~\ref{thm:ec}. Moreover, a zero-error code can be constructed that achieves this bound. Consider a block of length 3, denoted by $ b_1b_2b_3 $ at which the last bit serves as a parity check symbol such that $ b_3=b_1\oplus b_2$. Since at most one bit among these 3 bits may be erased, the erased bit can be recovered from other received bits. Therefore, the code yields zero-decoding error and, thus, $ C_0=2/3\approx 0.67 $.

According to Theorem~\ref{thm:finitem}, if  $ h_{lin}:=\log_2 |a|< C_0 $, bounded state estimation can be achieved. Moreover, if $ h_{lin}> C_0 $, no coder-estimator can estimate the state of the system. We consider two plants with different topological entropies. One of the plants has an entropy smaller than $ C_0 $, i.e., $ h_{lin}=C_0-\epsilon $, while the other one has a topological entropy larger than $ C_0 $, i.e., $ h_{lin}=C_0+\epsilon $ where $ \epsilon=0.1 $. We construct a coder-estimator that estimates the state with a known bounded estimation error. We use the achievable scheme from \cite[Ch. 8]{matveev2009estimation}. The details of the coding scheme are given in Appendix \ref{app:ne}.
%The operation of the encoder and decoder is organized into epochs $ \left[ \tau_i:=ir, \tau_{i+1} \right)  $, such that $ r=15 $. In each epoch, a $ 2^{C_0r} $-level quantizer $ \mathcal{Q}_r $ is used to transmit $ q = \mathcal{Q}_r(X) $ over a block of size $ r $. The decoder computes a state estimate $ \hat{X}(t) $. 
Fig. \ref{fig:NE} shows that the magnitude of estimation error $ E(\tau_i):=X(\tau_i)-\hat{X}(\tau_i) $ for both plants. Note that the error remains bounded for the plant with $ h_{lin}<C_0 $ but grows with time for the plant with $ h_{lin}>C_0 $.
%	\begin{align}
%	\|E(\tau_i)\| &< 2^{(h_{lin}-C_0)r} \bigg( \delta_0+ \frac{\delta_*}{1-2^{(h_{lin}-C_0)r}} \bigg)  + \frac{2^{h_{lin}r}-1}{2^{h_{lin}}-1}. \label{upz}
%	\end{align}

\section{State Estimation over Finite-state Additive Noise Channels} \label{sec:est_anc}
In this section, we investigate bounded state estimation of an LTI system over finite-state additive noise channels. The topological entropy of the channel appears in both necessary and sufficient conditions, reinforcing the links between topological entropies and the bounded state estimation.	
For finite-state additive noise channels, the conditions are as follows:

\begin{prop} %[Bounded estimation errors via finite-state additive noise channel]
	\label{prop:nssest} 
	Consider an LTI system in \eqref{lti1} satisfying conditions A1--A5. Assume that outputs are coded and estimated via a  finite-state additive noise channel~(Def. \ref{def:fanchannel}) with topological entropy $h_{ch}$. Then, uniformly bounded estimation errors can be achieved if
	\begin{align}
	h_{lin} + 2h_{ch} < 1 \label{2est1}.
	\end{align}
	Conversely, there exists a sequence of process and measurement noises for which the estimation error grows unbounded if
	\begin{align}
	h_{lin}+h_{ch} \label{2est2} > 1.
	\end{align}
\end{prop}
%	\begin{proof} The proof follows from Theorems \ref{thm:nssc0f}, \ref{thm:fanctp}, and \ref{thm:finitem}.  \end{proof}
\begin{rem}
	The inequality  $ h_{lin} + \alpha h_{ch} < 1 $ with  $ \alpha \in \{1,2 \} $ characterizes the bounded estimation condition: as long as the sum of uncertainties are small (below 1), we can achieve bounded estimation; otherwise, no such estimator can be constructed. Note that $ \alpha$ reflects the gap between achievability and converse.
\end{rem}
To prove {Proposition \ref{prop:nssest}, we first} investigate $ C_0 $ of the finite-state additive noise channel. We have the following result on bounding $ C_0 $ of these channels.
\begin{thm}\label{thm:fanctp}
	The zero-error capacity of a finite-state additive noise channel~(Def. \ref{def:fanchannel}) with topological entropy, $h_{ch}$, is bounded by
	\begin{align}
	1-2 h_{ch} \leq C_0 \leq 1-h_{ch}.
	\label{nsslmd}
	\end{align}
\end{thm}
\begin{proof}
	See Appendix \ref{app:nsslw}.
\end{proof}
\begin{figure*}
	\subcaptionbox{}{\includegraphics[width=65mm]{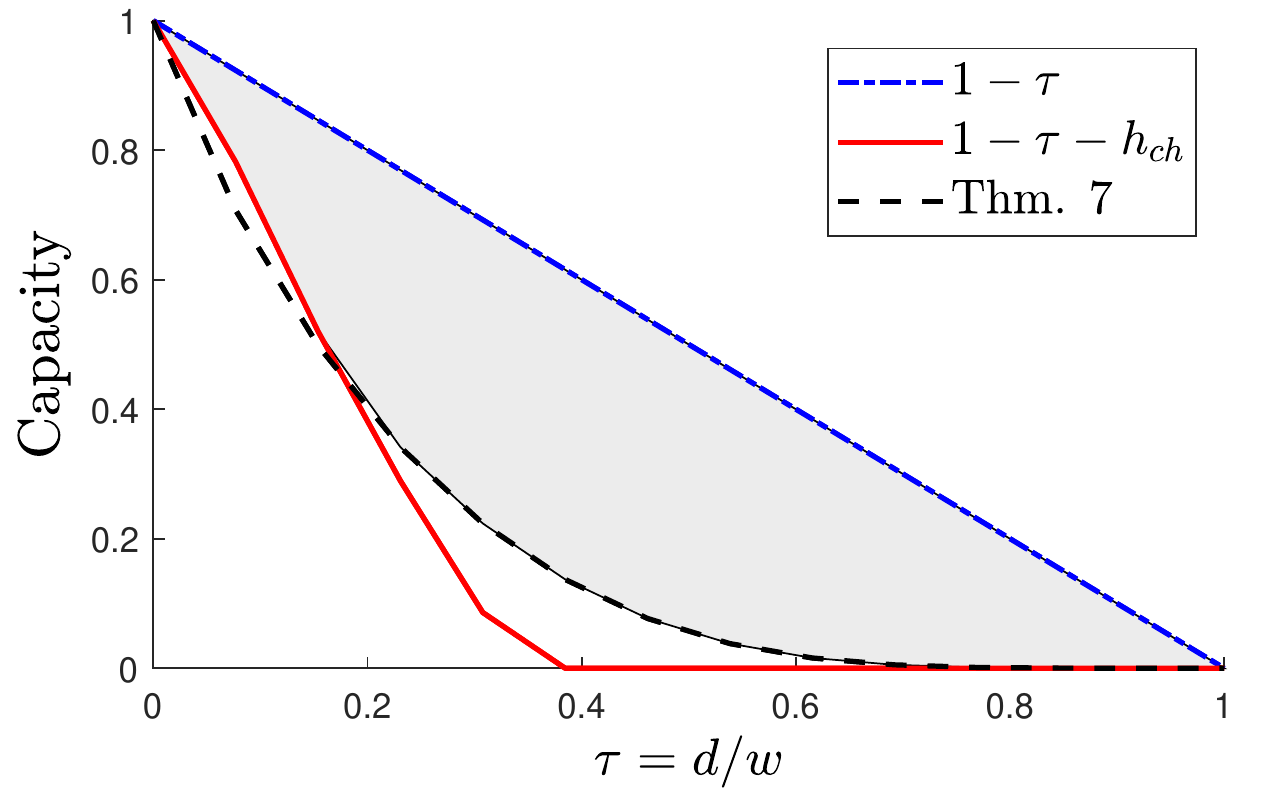}} \hspace{1mm}
	\subcaptionbox{}{\includegraphics[width=65mm]{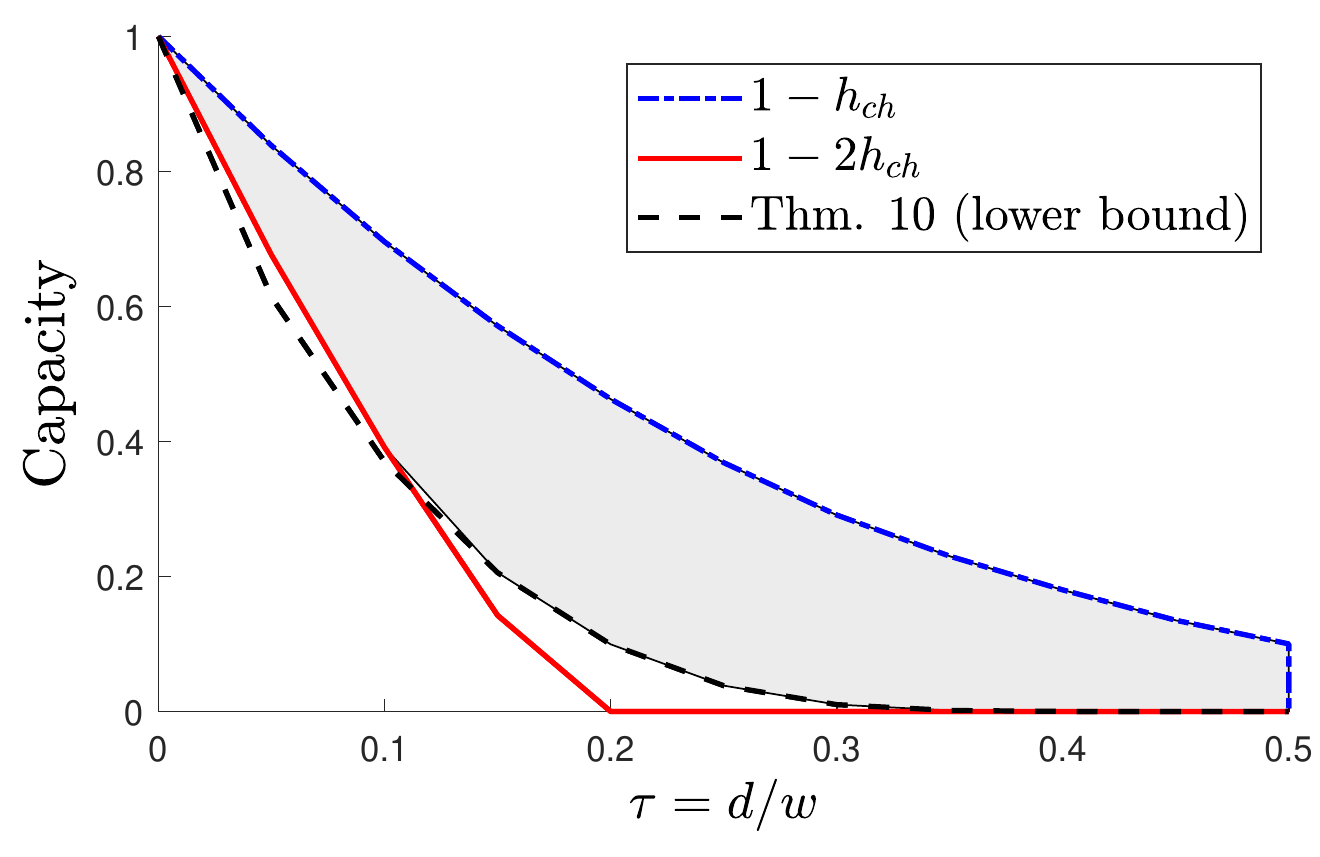}}
	\centering
	\caption{\small{Zero-error capacity bounds for (a) a binary sliding-window erasure channel and (b) a binary sliding-window symmetric channel. $ C_0 $ lies in the gray shaded area.}}
	\label{fig:cap_fse}
\end{figure*}
\begin{rem}
	The bounds in \eqref{nsslmd} show a clear relationship between zero-error capacity and topological entropy such that if there is a high uncertainty in the channel state transition graph, it leads to a linear reduction in zero-error capacity. In other words, if $ 1-C_0 $ considered as the {\em corruption rate} in the channel, then it is bounded as $ h_{ch} \leq 1-C_0 \leq 2h_{ch} $ which closely relates to the asymptotic growth rate of uncertainty in the channel.
\end{rem} 
\begin{rem}
	The proof of {Proposition} \ref{prop:nssest} follows from Theorems \ref{thm:finitem} and \ref{thm:fanctp}.
\end{rem}
In the sequel, the zero-error capacity of the sliding-window symmetric channel as an example of finite-state additive noise channel is discussed and compared with the general results of Theorem \ref{thm:fanctp}.
\subsection{Sliding-window symmetric channel}
In this subsection, the zero-error capacity of the sliding-window additive noise channel as an example of finite-state additive noise channel~(Section \ref{swcm}) is discussed and new bounds are derived.	
The following Theorem gives the bounds for the sliding-window symmetric channel.	
\begin{thm} \label{thm:nssc0f} The zero-error capacity of a sliding-window symmetric channel which have up to $ d $ errors in each sliding window of size $ w $ is bounded by
	\begin{align}
	1-\frac{1}{w}\log V^w_{2d}(q) \leq C_{0} \leq 1-\frac{d}{w} \log(q-1).\label{nssc0f}
	\end{align}
\end{thm}
Moreover, if $ d \geq w/2 $, then $ C_{0}=0 $.
		\begin{proof}
			See Appendix \ref{app:sws}.
		\end{proof}
		\begin{rem}
	The upper bound in \eqref{nssc0f} can be very loose for small alphabet sizes such that for a binary channel, it is the trivial bound of $ C_{0} \leq 1 $.
\end{rem}

The results of Theorems \ref{thm:ec}-\ref{thm:nssc0f} are summarized in Table I. Fig. \ref{fig:cap_fse} (b) shows the capacity bounds, for a binary sliding-window symmetric channel. Again, the general topological entropy lower bound in \ref{cor:swtb} gives a tighter lower bound for small $ d/w $. The upper bound in \eqref{nssc0f} in this case is equal to 1.	
\section{Conclusion} \label{sec:conclusions}
State estimation of linear time-invariant discrete-time systems over a class of finite-state channels was considered. The zero-error capacity of such channels is characterized in terms of nonstochastic information, leading to a necessary and sufficient condition for having uniformly bounded estimation errors. Two classes of finite-state channels are then described, namely finite-state erasure and additive noise channels. Bounds for the zero-error capacity of these channels were derived using results from feedback capacity and topological entropy theory. These bounds were translated to uniformly bounded state estimation over such channels. Interestingly, the results show that for finite-state uncertain channels, having strictly positive error-free communication rates, uniformly bounded estimation is possible. This contrasts sharply with the impossibility of  almost surely bounded estimation using standard stochastic models.

{Future work includes finding conditions in which a performance metric is guaranteed when estimating the states of a system over a channel. Moreover, Theorem \ref{thm:finitem} suggests that a larger range of unstable systems can be estimated with bounded errors over a channel modeled with memory compared to a memoryless worst-case version of the channel. Therefore, further analysis is needed regarding the modeling of channels for control purposes. Another direction is studying} the uniform stability of linear control systems via finite-state channels.

		\appendices
		\section{Proof of Theorem \ref{thm:fci}~($C_0$ via maximin information)} \label{app:cfi}
		By Def. \ref{def:c0}, $ {\mathscr{D}(n)}\subseteq \mathcal{X}^{n+1}$  is the set of all block codes of length $ n+1 $ that yield zero decoding errors for any channel noise sequence and initial state of the channel. In other words,
\begin{align}
\begin{split}
{\mathscr{D}(n)}= \big\{  \mathcal{F} \in  {\tilde{\mathscr{D}}(\mathcal{X}^{n+1})} : \forall  \, y&(0:n) \in \mathcal{Y}^{n+1},\\
&  | \mathcal{F}\cap \mathcal{G} (y(0:n)) |\leq 1 \big\} ,
\end{split} \label{zec}
\end{align}
where $ {\tilde{\mathscr{D}}(\mathcal{X}^{n+1})} $ is the family of all subsets of $ \mathcal{X}^{n+1} $ and where for each $y(0:n)\in\mathcal{Y}^{n+1}$,  the (set-valued) reverse transition function $ \mathcal{G}(y(0:n))$ gives the set of all sequences in $\mathcal{X}^{n+1}$ that could have produced an output $y(0:n)$. For any $y(0:n)\in \llbracket Y(0:n) \rrbracket$, we have 
\begin{align}
\llbracket X(0:n)|y(0:n)\rrbracket & 
= \llbracket X(0:n)\rrbracket \cap \mathcal{G}(y(0:n)). \label{revG}
\end{align} 
Since $ \llbracket X(0:n)|Y(0:n)\rrbracket_* $ is a partition of $ \llbracket X(0:n) \rrbracket $, thus choosing a single input in each partition, it can be used as a zero-error code. In other words,
\begin{align}
\big|\llbracket X(0:n&)|Y(0:n)\rrbracket_*\big| \nonumber\\
&=\max_{\mathcal{F} \in {\tilde{\mathscr{D}}(\mathcal{X}^{n+1})}: \forall \mathcal{C} \in \llbracket X(0:n)|Y(0:n)\rrbracket_*, |\mathcal{F} \cap \mathcal{C}|\leq 1}\,|\mathcal{F}|\nonumber \\
&\stackrel{(a)}{\leq} \max_{\mathcal{F} \in {\tilde{\mathscr{D}}(\mathcal{X}^{n+1})}: \forall \mathcal{B} \in \llbracket X(0:n)|Y(0:n)\rrbracket, |\mathcal{F} \cap \mathcal{B}|\leq 1}|\mathcal{F}|\nonumber \\
&\stackrel{\eqref{revG}}{=}\max_{\mathcal{F} \in {\tilde{\mathscr{D}}(\mathcal{X}^{n+1})}: \forall y(0:n) \in \llbracket Y(0:n) \rrbracket, |\mathcal{F} \cap \mathcal{G}(y(0:n))|\leq 1}|\mathcal{F}|\nonumber \\
&\stackrel{\eqref{zec}}{\leq} \quad \max_{\mathcal{F} \in {\mathscr{D}(n)}} |\mathcal{F}|\label{supf} \\
&\,\leq \quad q^{C_0(n+1)}. \nonumber 
\end{align}
Therefore, taking logarithm and dividing both sides by $ n+1 $,
\begin{align}
\frac{1}{n+1} I_*[X(0:n);Y(0:n)] \leq C_0,\quad \forall n \in \mathbb{N}_0. \label{minmaxinfoa}
\end{align}
Note that (a) follows from the fact that each partition is expressible by union of some $ \mathcal{B} = \llbracket X(0:n)|y(0:n)\rrbracket)$ \cite[Lemma 3.1]{nair2013nonstochastic}. % While, (b) follows from \eqref{revG}.% and the fact that any $ \mathcal{B} \in \llbracket X(0:n)|Y(0:n)\rrbracket \subseteq \llbracket X(0:n)\rrbracket$ are contained in $ \{\mathcal{G} (y(0:n))| y(0:n) \in \llbracket Y(0:n) \rrbracket \} $, because not all the output sequences in $ \mathcal{Y}^{n+1} $ might get observed.         
Next, it is shown that $ \forall n \in \mathbb{N}_0 $, there is an uncertain variable $ X(0:n) $ for which \eqref{supf} is an equality. 
Let $\mathcal{F}^*$ be a maximum-cardinality set in ${\mathscr{D}(n)}$, and set $\llbracket X(0:n)\rrbracket = \mathcal{F}^*$.
By construction for any zero-error code, no point in $ \mathcal{F}~=~\llbracket~X(~0~:~n~)~\rrbracket $ is $ \llbracket X(0:n)|Y(0:n)\rrbracket $-overlap connected. Thus the overlap partition $ \llbracket X(0:n)|Y(0:n) \rrbracket_* $ is a family of $ |\mathcal{F}^*|=|\llbracket X(0:n)\rrbracket | $ singletons, ensuring equality in \eqref{supf}. Taking logs and then a supremum over $n\in\mathbb{N}_0$ completes the proof.

		\section{Proof of Theorem \ref{thm:finitem}~(bounded estimation over finite-state uncertain channels)} \label{app:e}
		Suppose $ D(t)=\zeta (t,Y(0,t)) \in \mathcal{X}, t \in \mathbb{N}_0 $ be the channel's input where $ \zeta $ is an encoder operator. Each symbol $ D(t) $ is then transmitted over the channel. The received symbol $ Q(t) \in \mathcal{Y} $ is decoded and a causal prediction $ \hat{X}(t+1) $ of $ X(t+1) $ is produced by means of another operator $ \eta $ as $\hat{X}(t+1)=~\eta(t, Q(0:t)) \in \mathbb{R}^{n_x}, \,\hat{X}(0)=0$. We denote the estimation error as $ E(t):=X(t)-\hat{X}(t) $. 
		
		{\em 1) Necessity}:
		Assume a coder-estimator achieves uniform bounded estimation error. By change of coordinates, it can be assumed that $ A $ matrix is in {\em real Jordan canonical form} which consists of $ \varrho $ square blocks on its diagonal, with the $ j $-th block $ A_j \in \mathbb{R}^{n_j\times n_j} ,\, j=1,\dots,\varrho$. Let $ X_j(t),\hat{X}_j(t),E_j(t) \in \mathbb{R}^{n_j} $ and so on, be the corresponding $ j $-th component.
		Let $ \kappa \in \{1,\dots,n_x\} $ denote the number of eigenvalues with magnitude larger than $ 1 $, including repeats. From now on, we will only consider the unstable subsystem, as the stable part plays no role in the analysis. 
		By definition, uniformly bounded errors are achieved with any uniformly bounded $ \llbracket V(t) \rrbracket $ and $ \llbracket W(t) \rrbracket $, and any $ \llbracket X(0) \rrbracket $ contained in some $ \mathbf{B}_l $.  So set $\llbracket V(t)\rrbracket = \llbracket W(t)\rrbracket=\{0\}  $ and let $\llbracket X(0) \rrbracket $ be constructed as follows. Pick $
		\epsilon \in (0, 1-\max_{i:\lambda_i|>1} |\lambda_i|^{-1})$,
		arbitrary $ \nu \in \mathbb{N} $, and divide the interval $ [-l,l] $ on the $ i $-th axis into $ k_i:=\lfloor|(1-\epsilon)\lambda_i|^{\nu}\rfloor, \, i \in \{1,\dots,\kappa\}  $
		equal subintervals of length $ 2l/k_i $. Let $ p_i(s),\, s=\{1,\dots,k_i \}$ denote the midpoints of the subintervals and inside each subinterval construct an interval $ \mathbf{I}_i(s) $ centered at $ p_i(s) $ with a shorter length of $ l/k_i $. A hypercube family is defined as
		\begin{align}
			\mathscr{H}&=\bigg\{\prod_{i=1}^{\kappa} \mathbf{I}_i(s_i) :  s_i \in\{1,\dots,k_i\}, i \in \{1,\dots,\kappa\}\bigg\}, \label{hyper}
		\end{align}
		in which any two hypercubes are separated by a distance of $ l/k_i $ along the $ i $-th axis for each $ i \in \{1,\dots,\kappa\}$ (see Fig. \ref{fig:eq48}). Finally, set $  \llbracket X(0)\rrbracket =\cup_{\mathbf{L}\in \mathscr{H}} \mathbf{L} \subset \mathbf{B}_l \subset \mathbb{R}^{\kappa}$.
		\begin{figure}[t]
			\centering
			\includegraphics[width=50mm]{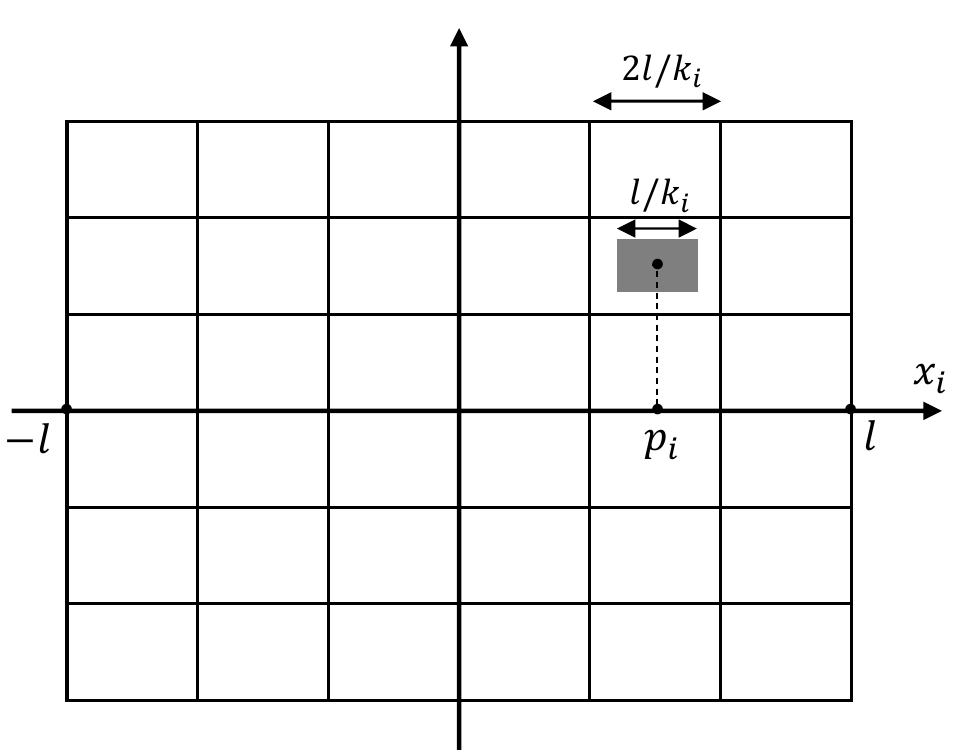}
			\caption{\small{The structure of each hypercube $ \mathbf{L} \in \mathscr{H} $ in \eqref{hyper}.}}
			\label{fig:eq48}
		\end{figure}	
		
		Let diam($ \cdot $) denote the set diameter under the $ l_{\infty} $ norm and given the received sequence $ q(0:t-1) $, we have
		\begin{align}
			\text{diam}  &\llbracket E_j(t)\rrbracket \geq \text{ diam} \llbracket E_j(t)| q(0:t-1)\rrbracket \label{condi}\\
			=&\text{ diam} \llbracket X_j(t)-\eta_j(t, q(0:t-1))| q(0:t-1)\rrbracket \nonumber\\
			\geq&  \text{ diam} \llbracket  A^t_jX_j(0)| q(0:t-1)\rrbracket \label{condi2}\\
			\geq& \, \sup_{u,v \in \llbracket  X_j(0)| q(0:t-1)\rrbracket}\frac{\norm{A^t_j(u-v)}_2}{\sqrt{n_x}}\nonumber\\
			\geq& \, \sup_{u,v \in \llbracket  X_j(0)| q(0:t-1)\rrbracket}\frac{\sigma_{min}(A^t_j)\norm{u-v}_2}{\sqrt{n_x}}\nonumber\\
			\geq& \,\, \sigma_{min}(A^t_j)\frac{\text{diam} \llbracket X_j(0)| q(0:t-1)\rrbracket}{\sqrt{n_x}}\label{diamf},
		\end{align}
		where $ \sigma_{min}(\cdot) $ denotes smallest singular value. \eqref{condi} holds since conditioning reduces the range \cite{nair2013nonstochastic}. Note that \eqref{condi2} follows from the fact that translating does not change the range. 	Using Yamamoto identity \cite[Thm. 3.3.21]{roger1994topics}, $ \exists t_{\epsilon} \in \mathbb{N}_0 $ such that $ \forall t\geq t_{\epsilon} $ the following holds
		\begin{align}
			\sigma_{min}(A^t_j) \geq (1-\frac{\epsilon}{2})^t |\lambda_{min}(A_j)|^t, \;j=1,\dots,p. \label{yamam}
		\end{align}
		By bounded state estimation error hypothesis $ \exists \phi>0  $, s.t. $ \phi \geq \sup \llbracket \norm{E(t)}\rrbracket
		\geq \sup \llbracket \norm{E_j(t)}\rrbracket\geq 0.5 \,\text{diam}\sup \llbracket E_j(t)\rrbracket $. Therefore,
		\begin{align}
			\phi \geq& \bigg((1-\frac{\epsilon}{2})|\lambda_{min}(A_j)|\bigg)^t\frac{\text{diam} \llbracket X_j(0)| q(0:t-1)\rrbracket}{2\sqrt{n_x}}.
		\end{align}
		Now, we show that for large enough $ \nu $, the hypercube family $ \mathscr{H} $	\eqref{hyper} is an $ \llbracket X(0)| Q(0:\nu-1)\rrbracket $-overlap isolated partition of $ \llbracket X(0)\rrbracket $. By contradiction, suppose that $ \exists \mathbf{L}\in \mathscr{H} $ that is overlap connected in $ \llbracket X(0)| Q(0:\nu-1)\rrbracket $ with another hypercube in $ \mathscr{H} $. Thus there exists a conditional range $ \llbracket X(0)| q(0:\nu-1)\rrbracket $ containing both a point $ u_j\in \mathbf{L} $ and a point $ v_j $ in some $ \mathbf{L}' \in \mathscr{H}\backslash \mathbf{L}$. Henceforth
		\begin{align*}
			\norm{u_j-v_j}\leq& \text{ diam}\llbracket X_j(0)| q(0:\nu-1)\rrbracket\nonumber\\
			\leq& \frac{2\sqrt{n_x}\phi}{((1-{\epsilon /2})|\lambda_{min}(A_j)|)^{\nu}} \;j=\{1,\dots,p\}, \nu \geq t_{\epsilon}.
		\end{align*}
		Notice that by construction any two hypercubes in $ \mathscr{H} $ are separated by a distance of $ l/k_i $, which implies
		\begin{align*}
			\norm{u_j-v_j} \geq \frac{l}{k_i} &=\frac{l}{\lfloor(1-\epsilon)|\lambda_i|\rfloor^{\nu}} \geq \frac{l}{|(1-\epsilon)|\lambda_{min}(A_j)||^{\nu}}
		\end{align*}
		The right hand side of this equation would exceed the right hand side of \eqref{yamam}, when $ \nu  $ is large enough that \[ \big( \frac{1-\epsilon/2}{1-\epsilon}\big)^{\nu}> 2\frac{\sqrt{n_x}\phi}{l}, \] yielding a contradiction.		
		Hence, for large enough $ \nu $, no two sets of $ \mathscr{H} $ are $ \llbracket X(0)| Q(0:\nu-1)\rrbracket $-overlap connected. So,
		\begin{align}
			I_*[X(0):&Q(0:\nu-1)]=\log|\llbracket X(0)| Q(0:\nu-1)\rrbracket_*|\nonumber\\
			&\geq \log |\mathscr{H}|\nonumber\\
			&=\log \bigg(\prod_{i=1}^{\kappa} k_i=\prod_{i=1}^{\kappa}\lfloor |(1-\epsilon)\lambda_i|^{\nu} \rfloor\bigg)\nonumber\\	
			&\geq \log \bigg( \prod_{i=1}^{\kappa} 0.5 |(1-\epsilon)\lambda_i|^{\nu}\bigg) \label{half}\\
			&=\nu \bigg( \kappa \log(1-\epsilon) -\frac{\kappa}{\nu}\log 2+\sum_{i=0}^{\kappa}\log|\lambda_i| \bigg),\label{istar}
		\end{align}
		where \eqref{half} holds since $ \lfloor x \rfloor > x/2, \forall x > 1$. Furthermore, condition A3 implies a Markov chain, i.e., $ X(0) \leftrightarrow S(0 :\nu) \leftrightarrow Q(0:\nu) $. Hence,
		\begin{align*}
			I_*[X(0);Q(0:\nu-1)] &\leq I_*[X(0:\nu-1);Q(0:\nu-1)] \\
			&\leq I_*[S(0:\nu-1);Q(0:\nu-1)] \\
			&\leq \nu C_0.
		\end{align*}
		Considering this and \eqref{istar} yields
		\begin{align*}
			C_0 \geq {\kappa}\log(1-\epsilon) + \sum_{i=0}^{\kappa}\log|\lambda_i|-\frac{\kappa}{\nu}\log 2.
		\end{align*}
		By letting $ \nu \rightarrow \infty $ and the fact that $ \epsilon $ can be made arbitrarily small yields $ C_0\geq h_{lin} $.
		
		{\em 2) Sufficiency}: 
		By \eqref{estdis} and \eqref{c0def},
		$ \forall \delta \in (0,C_0-h_{lin}), \, \exists t^*,\nu >0 \,(\nu> t^*) $ and a zero-error code book $ \mathcal{F} \subseteq\mathcal{X}^{\nu} $ such that
		\begin{align}
			h_{lin}<C_0 -\delta \leq \frac{1}{\nu} \log|\mathcal{F}|.\label{pc0}
		\end{align}
		Down-sample \eqref{lti1} by $ \nu $, the equivalent LTI system is
		\begin{align}
			X((k+1)\nu)&=A^{\nu}X(k\nu)+U'_{\nu}(k)+V'_{\nu}(k), \label{ltis}\\
			Y(k\nu)&=CX(k\nu)+W(k\nu), \, k\in \mathbb{N}_0, \label{ltio}
		\end{align}
		where the accumulated control term 
		\[U'_{\nu}(k)=\sum_{i=0}^{r} A^{\nu-1-i}BU(k\nu +i) \]
		and disturbance term 
		\[V'_{\nu}(k)=\sum_{i=0}^{r} A^{\nu-1-i}V(k\nu +i) \] 
		can be shown to be uniformly bounded over $  k \in \mathbb{N}_0 $ for each $ r \in [0:\nu-1] $. 
		Note that for a finite-state uncertain channel, a zero-error code $ \mathcal{F} $ can be used in consecutive blocks yielding no error at the decoder. This follows from considering all possible initial states of the channel, i.e., $ \llbracket S(0) \rrbracket =\mathcal{S} $ which implies $ \llbracket Y(k\nu:k\nu+r) |x(0:r)\rrbracket  \subseteq \llbracket Y(0:r) |x(0:r)\rrbracket$. This is because at the start of subsequent blocks $ \llbracket S(k\nu) \rrbracket \subseteq \mathcal{S} $.
		By \eqref{pc0}, $ |\mathcal{F}| $ codewords can be transmitted without error for which satisfies $ \log|\mathcal{F}|> \nu h_{lin} $. By the ``data rate theorem" for LTI systems with bounded disturbances over error-less channels (see e.g. \cite{tatikonda2004control2}) then there exists a coder-estimator for the equivalent LTI system of \eqref{ltis}-\eqref{ltio} with uniformly bounded estimation error for $ k \in \mathbb{N}_0 $.
		It readily gives the uniform boundedness for every $ t \in \mathbb{N}_0 $ of the \eqref{lti1}.
		
		\section{Proof of Lemma \ref{lem:maxer}}\label{app:a}
		Let $n$ be the length of the walk\footnote{The number of elements in a walk is called the {\em length} of the walk and denoted by $ \#\varpi$.} through the strongly connected state transition graph. If $ n>|\mathcal{S}| $ then according to the pigeonhole principle, there is at least one repeated vertex, and therefore, the walk must contain a cycle. Hence, any walk $ \varpi $ of length $n$ goes through some non-repeated vertices (forming an acyclic {\em path}) and some cycles. In other words, the walk contains a path and some cycles. The length of $ \varpi $ is the sum of the two sublists length, i.e.,
		\begin{align}
			n =\#\mathbf{\varpi} &= \#\varpi_{nc} + \# \varpi_c, \label{wakc}
		\end{align}
		where $ \varpi_{nc} $ is the list of edges passing through non-repeated vertices and $ \#\varpi_c := \sum_{i}\#\varpi_c(i)$ is the sum of the length of all the visited cycles (including repeated cycles) indexed by $ i $. A sample walk is shown in Fig. \ref{fig:walks} where any cycle is simplified by a self-loop. 
		The walk length through non-repeated vertices can not exceed the total number of vertices, i.e. $\#\varpi_{nc} \leq |\mathcal{S}|$.	
		The reason is that if there was a state visited twice then the whole path can be considered as a cycle. As an example, in Fig. \ref{fig:walks}, $ s_i=s_j $ then the walk would be another cycle (colored blue). Since length of the walk associated with cycles, as a sublist is smaller than the whole walk, we have
		\begin{align}
			n =\# \varpi 	&= \# \varpi_{nc} + \# \varpi_c \geq \# \varpi_c = \sum_{i}\#\varpi_c(i), \label{cwak}
		\end{align}
		where the sum is over all the cycles in the walk. Next, we show that the number of erasure-edges in every walk is upper bounded by $E(\varpi) \leq \tau n+ |\mathcal{S}|$.
		For walks through non-repeated vertices, we have $ E(\varpi_{nc}) \leq \#\varpi_{nc} = |\mathcal{S}| $ (visited erasure edges are a sublists of the walk). For walks within cycles, 
		\begin{align}
			E(\varpi_c) &= \sum_{i} \#\varpi_c(i) \tau_i \leq \tau \sum_{i} \#\varpi_c(i) \leq \tau n.	\label{5-4}	
		\end{align}
		The first inequality follows by Def. \ref{def:mr}. Therefore, considering \eqref{wakc} and \eqref{5-4} we get $E(\varpi) = E(\varpi_{nc})+E(\varpi_c) \leq \tau n + |\mathcal{S}|$.
		\begin{figure}[t]
			\centering
			\begin{tikzpicture}[->, >=stealth', auto, semithick, node distance=5cm]
				\tikzstyle{every state}=[fill=white,draw=black,thick,text=black,scale=0.3]
				\tikzstyle{surround} = [circle,fill=cblue!70,draw=black]
				\node[state]   (S1)       {};
				\node[state]    (S2)[right of=S1]   {};
				\node[state]    (S3)[right of=S2, label=below:$s_i$]   {};
				\node[state]    (S4)[right of=S3]   {};
				\node[state]    (S5)[right of=S4, label=below:$s_j$]   {};
				\node[state]    (S6)[right of=S5]   {};
				\node[state]    (S7)[above right of=S2]   {};
				\node[state]    (S8)[above left of=S2]   {};
				\node[state]    (S9)[above left of=S7]   {};
				\node[state]    (S10)[above right of=S4]   {};
				\node[state]    (S11)[above left of=S4]   {};
				\node[state]    (S12)[above left of=S10]   {};
				\path
				(S1) edge[] node [right] {} (S2)
				(S2) edge[bend right,looseness=1] node [] {} (S7)
				edge[] node [right] {} (S3)
				(S3) edge[] node [] {} (S4)
				(S4) edge[bend right] node [right] {} (S10)
				edge[] node [right] {} (S5)
				(S5) edge[] node [right] {} (S6)
				%(S5) edge[bend right,red] node [] {} (S3)
				(S7) edge[bend right, dotted] node [] {} (S9)
				(S9) edge[bend right] node [] {} (S8)
				(S8) edge[bend right] node [] {} (S2)
				(S10) edge[bend right, dotted] node [] {} (S12)
				(S12) edge[bend right] node [] {} (S11)
				(S11) edge[bend right] node [] {} (S4);
				\begin{pgfonlayer}{background}
					\node[surround,yshift=22pt,looseness=20,scale=.9] [fit =(S3)(S4)(S5)] {};
				\end{pgfonlayer}
			\end{tikzpicture}
			\caption{\small{Combination of non-repeated vertices and cycles in a walk through a strongly connected graph. All non-repeated states are distinct, otherwise, another cycle exists that include all the walk done in between; e.g. if in this graph $ s_i=s_j $, then the blue circle corresponds to a bigger cycle, containing the walk starting from $ s_i $ and ending in $ s_j $. }}
			\label{fig:walks}
		\end{figure}
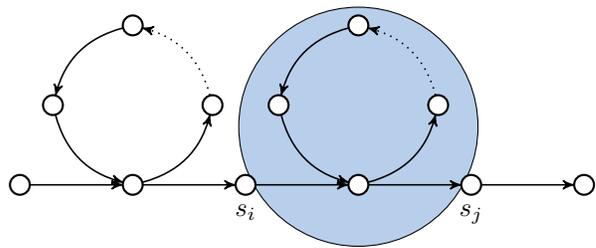
	
		\section{Proof of Lemma \ref{lem:miner}} \label{app:b}
		Consider the set of states denoted by $ \mathcal{S}_m $, that build the cycle with maximal ratio $ \tau=e/l $; e.g. in Fig. \ref{fig:transition}, $ \mathcal{S}_m=\{s_2,s_3,s_4\} $ shapes the cycle with maximal ratio. Any walk on this cycle has at least $ \lfloor{n/l}\rfloor e $ erasure edges. In other words,
		\begin{align*}
		E(s_0,\mathbf{y}) \geq \lfloor{n/l} \rfloor e > \tau n - l,\, s_0 \in \mathcal{S}_m. 
		\end{align*}
		
		Furthermore, since in a strongly connected graph any state is reachable within $ |\mathcal{S}| $ steps, then starting from $\forall s_0 \in \mathcal{S}$ we have $E(s_0,\mathbf{y}) > \tau n - l - |\mathcal{S}|$.
		\section{Proof of Theorem \ref{thm:nsec0f}~($ C_{0f} $ of finite-state erasure channel)} \label{app:c}  
		The set of states, $ \mathcal{S} $ of the channel can be partitioned into two subsets $ \mathcal{S}_I $ and $ \mathcal{S}_{II} $, where, $ \mathcal{S}_I $ is the set of states that have two outgoing edges, one error-free transmission and the other one with erasure, and $ \mathcal{S}_{II} $ is the set of states that have a single error-free outgoing edge; e.g. in Fig. \ref{fig:transition}, $ \mathcal{S}_I=\{s_1,s_4\}  $ and $ \mathcal{S}_{II}=\{s_2,s_3\} $. {We obtain $ C_{0f} $ of the channel using the result of \cite{zhao2010zero}, in which it is shown that
	\begin{align}
	C_{0f}= \liminf_{k \rightarrow \infty} \frac{1}{k} \min_{s \in \mathcal{S}} \log W(k,s)\label{c0Fs},
	\end{align}
	where $ \forall s \in \mathcal{S} $ and for $ k \in \mathbb{N},\, W(k,s) $ is a mapping $ \mathbb{Z}^+ \times \mathcal{S} \mapsto \mathbb{R}^+ $ and is obtained iteratively (with initial value $ W(0,s)=1, \forall s \in \mathcal{S} $) from the dynamic programming (DP) equation in
	\begin{align}
	\begin{split}
	W(k,s)= \max_{P_{X|S}} &\min_{s' \in \mathcal{S}} \Bigg\{ W(k-1,s') \\ \times &\bigg(\max_{y \in \mathcal{Y}} \sum_{x \in \mathcal{G}(y,s'|s)}P_{X|S}(x|s)\bigg)^{-1} \Bigg\} \label{wit},
	\end{split}
	\end{align}
	with $ P_{X|S}(\cdot|\cdot) $ being a probability mass function on $\mathcal{X} $ for each state  $ s \in \mathcal{S} $. The subset of the inputs that can result in the output $y$ is denoted by $  \mathcal{G}(y,s'|s)=\{x| x \in \mathcal{X}, P_X(y,s'|x,s)>0\} $, where $ s $ is the current state and $ s' $ is the next state of the channel.
	For the finite-state erasure channel, $ \mathcal{G}(y,s'|s) =\{y\} $ if no erasure occurs (hence, $ x=y $) and $\mathcal{G}(y,s'|s)= \mathcal{X} $, otherwise.}		
Assume that the current state $s \in \mathcal{S}_{I}$, then there is two outgoing edges. {Let $ s_e $ and $ s_s $ be the  end-point of erasure and error-free edges, respectively. Hence}

\begin{align}
\begin{split}
W(k,s)&=\max_{P_{X|S}} \min\bigg\{W(k-1,s_e)\big(\max_{y \in \mathcal{Y}} \sum_{x \in \mathcal{X}} P_{X|S}(x|s)\big)^{-1}\\
&\qquad , W(k-1,s_s)\big(\max_{y \in \mathcal{Y}} P_{X|S}(x=y|s)\big)^{-1}\bigg\}
\end{split}\nonumber\\
\begin{split}
&\stackrel{(a)}{=} \min\bigg\{W(k-1,s_e)\times 1, \\ &\qquad W(k-1,s_s)\max_{P_{X|S}} \big(\max_{y \in \mathcal{Y}} P_{X|S}(x=y|s)\big)^{-1}\big\}\end{split} \nonumber \\
&= \min\bigg\{ W(k-1,s_e), W(k-1,s_s)\times q \bigg\}, \forall s \in \mathcal{S}_{I}. \label{its3}
\end{align}	
Note that $\sum_{x \in \mathcal{X}}P_{X|S}(x|s)$ is equal to $1$, as the summation is on all input alphabet. Furthermore, because one of the elements (i.e. $W(k-1,s_e)$) is constant w.r.t. the input distribution in the minimization argument, then the max-min operation can be swapped to get (a). At last, the uniform distribution is the solution of
$\max_{P_{X|S}}\big(\max_{y \in \mathcal{Y}}P_{X|S}(x=y|s)\big)^{-1}$ 
which equals $q$ and gives \eqref{its3}. Note that, \eqref{its3} shows the edge with erasure, multiplies a gain of 1. Whereas, the edge with error-free transmission multiplies a gain of $ q $. Furthermore, if current state $ s\in \mathcal{S}_{II} $, it leads to	
\begin{align}
W(k,s)&= \max_{P_{X|S}} \Bigg\{W(k-1,s') \big(\max_{y \in \mathcal{Y}} P_{X|S}(x=y|s)\big)^{-1} \Bigg\}\nonumber\\
&= W(k-1,s') \times q,\; \forall s \in \mathcal{S}_{II}. \label{its4}
\end{align}
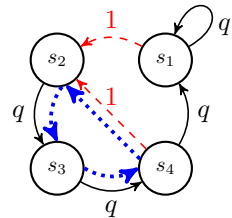
\begin{figure}[t]
	\centering
	\begin{tikzpicture}[->, >=stealth', auto, semithick, node distance=1.8cm]
	\tikzstyle{every state}=[fill=white,draw=black,thick,text=black,scale=0.8]
	\node[state]   (S1)       {$  s_1 $};
	\node[state]    (S2)[left of=S1]   {$ s_2 $};
	\node[state]    (S3)[below of=S2]   {$ s_3 $};
	\node[state]    (S4)[right of=S3]   {$ s_4 $};
	\path
	(S1) edge[out=30,in=70,looseness=8] node [below right] {$ q $} (S1)
	edge[dashed, red,bend right]  node [above] {$ 1 $} (S2)
	(S2) edge[bend right]  node [left] {$ q $}  (S3)
	(S2.280) edge[blue, dotted, bend right, line width=.5mm]   (S3.90)
	(S3) edge[bend right]  node [below] {$ q $} (S4)
	(S3.0) edge[blue, dotted, bend right, line width=.5mm]   (S4.180)
	(S4) edge[bend right]  node [right] {$ q $} (S1)
	(S4) edge[dashed, red]  node [above] {$ 1 $} (S2)
	(S4.160) edge[blue, dotted, line width=.5mm]   (S2.290);
	\end{tikzpicture}
	\caption{\small{The gain $a(s,s') \in \{1,q\}$ associated with each edge in a $ (w=3,d=1) $ sliding-window erasure channel. The cycle shaped with the states obtained from the solution of DP problem for $ s \in \mathcal{S}_m $ is highlighted with a blue-dotted line.}}
	\label{fig:transition2}
\end{figure}
From \eqref{its3} and \eqref{its4} it yields that at each iteration, a gain is multiplied to the cost-to-go function, i.e., $ W(k-1,s') $. We denote this gain with $ a(s,s') $. This gain is obtained by solving $ \max_{P_{X|S}} (\max_{y \in \mathcal{Y}} \sum_{x \in \mathcal{G}(y,s'|s)} P_{X|S}(x|s))^{-1} $ for given $ s $ and $ s' $ that results in
\begin{align}
a(s,s')=
\begin{cases}
1 & \text{if } \mathcal{G}(y,s'|s)=\mathcal{X},
\\
q & \text{if } \mathcal{G}(y,s'|s)=y.
\end{cases} \label{gain}
\end{align}
In other words, if going from $ s $ to $ s' $ is a path or edge with erasure (or when $ \mathcal{G}(y,s'|s)= \mathcal{X}$) the gain is $a(s,s')=1 $ and if it is for error-free edge (or when $ \mathcal{G}(y,s'|s)= y$) then $a(s,s')=q $.
In Fig. \ref{fig:transition2} the associated gain for each edge is shown. The red-dashed lines that represent erasure have a gain of $ 1 $ and other edges which represent error-free transmissions have a gain of $ q $.

Therefore, starting from any initial state, solving the DP problem of \eqref{wit} for finite-state erasure channel, corresponds to the output sequence with the maximum number of erasures that gives minimum overall gain, i.e. assuming $ s(0)=s $,
\begin{align}
W(k,s)&=\min_{s(0:k-1) \in \mathscr{S}(k)} \prod_{i=0}^{k-1} a(s(i),s(i+1)), \label{gain2}
\end{align}	
where $ \mathscr{S}(k) $ denotes the set of all the state sequences of length $ k $ subject to the state transition graph $ \mathscr{G} $. 

For calculating $ C_{0f} $ using \eqref{c0Fs}, we need to find the output sequence with maximum number of erasures. According to \eqref{gain2}, \eqref{maxeras} and \eqref{worsteras}, for any initial state $ s $ and $ k $ number of steps, we have
$q^{k-\tau k - |\mathcal{S}|} <  W(k,s)< q^{k-\tau k + l + |\mathcal{S}|}$
By taking a logarithm and dividing it by $n$,
\begin{align*}
1-\tau -\frac{l}{k} <& \frac{1}{k}\log W(k,s)< 1-\tau +\frac{l+|\mathcal{S}|}{k}.
\end{align*}
When $ k \rightarrow \infty $ the upper and lower bounds meet in $ 1-\tau $ which gives $	C_{0f}=\lim_{k \rightarrow \infty} \frac{1}{k} \min_{s \in \mathcal{S}} \log W(k,s) =1-\tau.$
		
		\section{Proof of Theorem  \ref{thm:ec}~(finite-state erasure channel topological bound)} \label{app:nselw}
		Let $ \mathbf{c}_1 \in \mathcal{X}^n$ be the first codeword for which adjacent inputs denoted by $ \mathscr{Q}(\mathbf{c}_1) $ depend on the number and position of erasures of each output sequence in $\mathscr{Y}_T(\mathbf{c}_1):= \cup_{s_0 \in \mathcal{S}} \mathscr{Y}(s_0,\mathbf{c}_1) $. Let $\mathcal{G}(\mathbf{y})$ denotes the set of input sequences that produce output sequence $ \mathbf{y}=y(0:n-1) $. Hence,
		\begin{align}
		\mathscr{Q}(\mathbf{c}_1)&=\bigcup_{\mathbf{y} \in \mathscr{Y}_T(\mathbf{c}_1) } \mathcal{G}(\mathbf{y}). \label{q1}
		\end{align}
		Also, we have $\mathcal{G}(\mathbf{y}) :=\bigcup_{s_0 \in \mathcal{S}} \mathcal{G}(s_0,\mathbf{y})$, where $ \mathcal{G}(s_0,\mathbf{y}) $ is the subset of inputs that can result in output $\mathbf{y}$ with initial state $ s_0 $.
		Therefore, $|\mathcal{G}(\mathbf{y})| \leq \sum_{s_0 \in \mathcal{S}} |\mathcal{G}(s_0,\mathbf{y})|$, which gives,
		\begin{align}
		|\mathscr{Q}(\mathbf{c}_1)|&\leq \sum_{\mathbf{y} \in \mathscr{Y}_T(\mathbf{c}_1)}\sum_{s_0 \in \mathcal{S}} |\mathcal{G}(s_0,\mathbf{y})|. \label{qsz}
		\end{align}
		Next, we give an upper bound on $ |\mathcal{G}(s_0,\mathbf{y})| $.
		\begin{lem} \label{lem:gsize}
			For the finite-state erasure channel, the following inequality holds
			\begin{align}
			|\mathcal{G}(s_0,\mathbf{y})| \leq \gamma q^{\tau n}, \label{gsz}
			\end{align}
			where $ \gamma >0 $ is a constant number.
		\end{lem}
		\begin{proof}
			Number of erasures in each output sequence determines the size of $ \mathcal{G}(s_0,\mathbf{y}) $. Henceforth, using the result of Lemma \ref{lem:maxer}, the set of inputs that can produce any output is upper bounded by $|\mathcal{G}(s_0,\mathbf{y})| \leq q^{\tau n+ |\mathcal{S}|}= \gamma q^{\tau n}$, where $ \gamma=q^{|\mathcal{S}|} $.
		\end{proof}
		Therefore, substituting \eqref{gsz} into \eqref{qsz} yields
			\begin{align*}
			|\mathscr{Q}(\mathbf{c}_1)|	\leq \sum_{\mathbf{y} \in \mathscr{Y}_T(\mathbf{c}_1)} \sum_{s_0 \in \mathcal{S}} \gamma q^{\tau n}&\leq \sum_{\mathbf{y} \in \mathscr{Y}_T(\mathbf{c}_1)} |\mathcal{S}| \times \big(\gamma q^{\tau n}\big)\\
			&= | \mathscr{Y}_T(\mathbf{c}_1)| |\mathcal{S}| \times \big(\gamma q^{\tau n}\big).
			\end{align*}
			Because $\mathscr{Y}_T(\mathbf{c}_1):= \cup_{s_0 \in \mathcal{S}} \mathscr{Y}(s_0,\mathbf{c}_1) $. According to \eqref{outlam}, for any initial state the number of outputs is upper-bounded by $ \beta \lambda^n $. Therefore,
		\begin{align*}
		|\mathscr{Q}(\mathbf{c}_1)|&\leq \big(|\mathcal{S}| (\beta \lambda^n) \big)(\gamma |\mathcal{S}| q^{\tau n}) =\gamma \beta|\mathcal{S}|^2 (q^{\tau}\lambda)^n.
		\end{align*}
		Choosing non-adjacent inputs as the codebook results in error-free transmission. The above argument is true for other codewords, i.e., $|\mathscr{Q}(\mathbf{c}_i)| \leq \zeta \times (q^{\tau}\lambda)^n, i \in \{1,\dots, M \},$ where, $ \zeta:=\gamma \beta|\mathcal{S}|^2 $  and $ M $ is the number of codewords in the codebook constructed with above method such that $\mathcal{X}^n = \bigcup_{i=1}^M \mathscr{Q}(\mathbf{c}_i)$.\footnote{{This holds by construction. If $\mathcal{X}^n \neq \bigcup_{i=1}^M \mathscr{Q}(\mathbf{c}_i)$, then $ \exists x_{1:n} \in$ \mbox{$ \mathcal{X}^n  \setminus \bigcup_{i=1}^M \mathscr{Q}(\mathbf{c}_i) $} . We can choose $ x_{1:n} $ as another codeword (it is non-adjacent to other codewords).}}
		Considering the size of these sets, we have
		\begin{align*}
		q^n = |\bigcup_{i=1}^M \mathscr{Q} (\mathbf{c}_i)| \leq \sum_{i=1}^{M} |\mathscr{Q} (\mathbf{c}_i)|\leq M 	\times \big( \zeta (q^{\tau}\lambda)^n \big).
		\end{align*}
		Therefore, the number of distinguishable inputs is lower bounded by $M \geq q^n/(\zeta (q^{\tau}\lambda)^n)$.
		Hence, it builds a lower bound for the zero-error capacity
		\begin{align}
		C_0 \geq\frac{\log M}{n} &\geq 1-\tau-\log \lambda - \frac{\log \zeta}{n}.
		\end{align}
		If $n$ is large last term vanishes and result the lower bound in \eqref{ecc0}. Note that the result for $ C_{0f} $ in Theorem \ref{thm:nsec0f} builds an upper bound for the zero-error capacity (without feedback), hence, the bounds on the zero-error capacity of finite-state erasure channel in \eqref{ecc0} hold.
		
		\section{Proof of Lemma \ref{lem:tec}} \label{app:tec}
		In the sliding window erasure channel for every $ w $ transmissions at most $ d $ erasures can happen. Therefore, having $ n > w $ transmissions the number of erasures can not exceed $ n(d/w) + w $. Therefore, for any walk $ \varpi(0:n-1) $ through the associated graph, we have 
		\[\frac{E(\varpi(0:n-1))}{n} \leq \frac{d}{w} +\frac{w}{n}.\]
		
		On the other hand, starting from error-free state, i.e., no erasure in past $ w $ transmission and having $ d $ erasures periodically, will result in 
		\[\frac{E(\varpi(0:n-1))}{n} \geq \frac{d}{w} +\frac{d}{n}.\] 
		Therefore, 
		\begin{align}
		\lim_{n \rightarrow \infty} \max_{\varpi(0:n-1)} \frac{E(\varpi(0:n-1))}{n} = \frac{d}{w}. \label{limdn} 
		\end{align}
		Comparing \eqref{limdn} with \eqref{ratio} results in $ \tau = d/w $.
		\section{Proof of Theorem \ref{thm:fanctp}~(finite-state additive noise channel topological bounds)} \label{app:nsslw}
		The lower and upper bounds are proven in the sequel subsections, separately.
		\subsection{The lower bound}
		First, we give the following Lemma.
		\begin{lem} \label{lemnss}
			Let $ \mathcal{G}(s_0,\mathbf{y}) $ be subset of the inputs that can result in output $\mathbf{y}=y(0:n-1)$ with initial state $ s_0 $ for the finite-state additive noise channel. The following holds
			\begin{align}
			\alpha \lambda^n \leq |\mathcal{G}(s_0,\mathbf{y})| \leq \beta \lambda^n,
			\end{align}
		\end{lem}
		where $ \alpha $ and $ \beta $ are constants appeared in \eqref{outlam}.
		\begin{proof}
			The output sequence, $ \mathbf{y} $, is a function of input sequence, $ \mathbf{x}=x(0:n-1) $, and channel noise, $ \mathbf{v}=v(0:n-1) $, which can be represented as the following
			\begin{align}
			\mathbf{y} = \mathbf{x} \oplus \mathbf{v},  \label{adnoise}
			\end{align}
			where $ \mathbf{v} \in \mathcal{V}(s_0,n) $. The set of all output sequences $ \mathscr{Y}(s_0,\mathbf{x}) $ can be obtained as $\mathscr{Y}(s_0,\mathbf{x})=\{ \mathbf{x} \oplus \mathbf{v}| \mathbf{v} \in \mathcal{V}(s_0,n)\}$.
			Since for given $ \mathbf{x} $, \eqref{adnoise} is bijective, we have the following
			\begin{align}
			|\mathscr{Y}(s_0,\mathbf{x})|=| \mathcal{V}(s_0,n)|. \label{yvsize}
			\end{align}
			On the other hand, the subset of the inputs that can result in output $\mathbf{y}$ with initial state $ s_0 $, $ \mathcal{G}(s_0,\mathbf{y}) $ is defined as $\mathcal{G}(s_0,\mathbf{y})=\{ \mathbf{x} |\mathbf{x} \oplus \mathbf{v} = \mathbf{y}, \mathbf{v} \in \mathcal{V}(s_0,n)\}.$		
			Again, fixing $ \mathbf{y} $, the mapping $\mathbf{x} \to \mathbf{v}$ in \eqref{adnoise} is bijective, hence $|\mathcal{G}(s_0,\mathbf{y})|=| \mathcal{V}(s_0,n)|$. Combining it with \eqref{yvsize} yields $ |\mathcal{G}(s_0,\mathbf{y})|= |\mathscr{Y}(s_0,\mathbf{x})| $. Moreover, Proposition \ref{thm:out} gives the bounds on $ |\mathscr{Y}(s_0,\mathbf{x})| $. 
		\end{proof}
		
		Similar to Appendix \ref{app:nselw}, let $ \mathbf{c}_1 \in \mathcal{X}^n $ be the first codeword for which adjacent inputs denoted by $ \mathscr{Q}(\mathbf{c}_1) $. Again, each output sequence is in $\mathscr{Y}_T(\mathbf{c}_1):= \cup_{s_0 \in \mathcal{S}} \mathscr{Y}(s_0,\mathbf{c}_1) $. Hence,
		\begin{align}
		\mathscr{Q}(\mathbf{c}_1)&=\bigcup_{\mathbf{y} \in \mathscr{Y}_T(\mathbf{c}_1) } \mathcal{G}(\mathbf{y}), \label{q3}
		\end{align}
		where, $ \mathcal{G}(\mathbf{y}):=\bigcup_{s_0 \in \mathcal{S}} \mathcal{G}(s_0,\mathbf{y}) $,
			which gives 
			\begin{align*}
			|\mathscr{Q}(\mathbf{c}_1)|\leq\quad& \sum_{\mathbf{y} \in \mathscr{Y}_T(\mathbf{c}_1)}\sum_{s_0 \in \mathcal{S}} |\mathcal{G}(s_0,\mathbf{y})|\\
			\stackrel{\text{Lem.}\,\ref{lemnss}}{\leq}\,& \sum_{\mathbf{y} \in \mathscr{Y}_T(\mathbf{c}_1)} \sum_{s_0 \in \mathcal{S}} \beta \lambda^{ n}\\
			=\quad& \sum_{\mathbf{y} \in \mathscr{Y}_T(\mathbf{c}_1)} |\mathcal{S}| \times \beta \lambda^{ n}\\
			=\quad& | \mathscr{Y}_T(\mathbf{c}_1)|  \times \big(|\mathcal{S}| \beta \lambda^{ n}\big)\\
			=\quad& |  \cup_{s_0 \in \mathcal{S}} \mathscr{Y}(s_0,\mathbf{c}_1)|  \times \big(|\mathcal{S}| \beta \lambda^{ n}\big)	\\
			\stackrel{\text{Thm.}\,\ref{thm:out}}{\leq}&\, |\mathcal{S}|\big(\beta \lambda^{ n}\big)  \times \big(|\mathcal{S}| \beta \lambda^{ n}\big)= \big( \beta|\mathcal{S}| \lambda^n \big)^2.
			\end{align*}
		Again, similar to the proof of finite-state erasure channel in Appendix \ref{app:nselw}, choosing non-adjacent inputs as the codebook results in error-free transmission. The above argument is true for other codewords, i.e., $
		|\mathscr{Q}(\mathbf{c}_i)|\leq \big( \beta|\mathcal{S}| \lambda^n \big)^2, i \in \{1,\dots, M \}$,where $ M $ is the number of codewords in the codebook such that union of corresponding $\mathscr{Q}(\mathbf{c}_i)$ for $i=1,\dots,M,$ covers $ \mathcal{X}^n $. Then,
		\begin{align*}
		q^n = |\bigcup_{i=1}^M \mathscr{Q} (\mathbf{c}_i)| &\leq \sum_{i=1}^{M} |\mathscr{Q} (\mathbf{c}_i)|\leq M 	\times \big( \beta|\mathcal{S}| \lambda^n \big)^2.		
		\end{align*}
		As a result, the number of distinguishable inputs is lower bounded by $M \geq q^n/( \beta|\mathcal{S}| \lambda^n )^2$.
		Therefore, according to zero-error capacity definition
		\begin{align*}
		C_0&\geq \frac{1}{n}\log \frac{q^n}{(\beta |\mathcal{S}|)^2 (\lambda)^{2n}} =1-2\log \lambda - \frac{2}{n} \log (\beta |\mathcal{S}|).
		\end{align*}
		If $n$ is large, the last term vanishes and proves the lower bound in \eqref{nsslmd}.
		\subsection{The upper bound} 
		We prove the upper bound in \eqref{nsslmd}. Here, we show a more general result that holds for zero-error feedback capacity, which itself is an upper bound for $ C_0 $.
		
		We use a similar idea used in \cite{ahlswede1993nonbinary} and \cite{ahlswede2006non} to derive the upper bound.
		Let $ m \in \mathcal{M} $ be the message to be sent and  $ \mathbf{y}=y(0:n) $ be the output sequence such that 
		\begin{align*}
			y(t)=f_{t}(m,y(0:t-1))\oplus v(t), \, t \in \{0,1,\dots, n\},
		\end{align*}
		where $ \mathbf{v}=v(0:n) \in \mathcal{V}(s_0,n) \in \mathcal{X}^{n+1}$ is the additive noise and $ f_{t}(m,\cdot) $ the encoding function. Therefore, the output is a function of encoding function and noise sequence, i.e., $ \mathbf{y}=\psi(f_{0:n}(m,y(0:n-1)),\mathbf{v}) $.
		Let the family of encoding functions $ \mathcal{F}=\{f_{0:n}(m, \cdot)| m \in \mathcal{M}\} $. We denote all possible outputs $ \Psi (\mathcal{F},\mathcal{V}(s_0,n))=\{\mathbf{y}| m \in \mathcal{M}, \mathbf{v} \in \mathcal{V}(s_0,n) \} $. For having a zero-error feedback code any two $ m, m' \in \mathcal{M}, m\neq m'$ and any two $ \mathbf{v}, \mathbf{v}' \in \mathcal{V}(s_0,n)$ must result in $ \psi(f_{0:n}(m,y(0:n-1)),\mathbf{v})\neq \psi(f_{0:n}(m',y'(0:n-1)),\mathbf{v}') $. Note that when $ m = m' $, (even with feedback) at first position that $ v(t) \neq v'(t), t\in \{0,1,\dots, n\},$ will result in $  y(t) \neq y'(t)$. Therefore, assuming the initial condition is known at both encoder and decoder,
		\begin{align*}
			|\Psi (\mathcal{F},\mathcal{V}(s_0,n))| =M|\mathcal{V}(s_0,n)| \leq q^{n+1}.
		\end{align*}
		Therefore, $ M $ is an upper bound on the number of messages that can be transmitted when the initial condition is not available, i.e. $|\mathcal{M}| \leq M$. From  Def. \ref{def:c0f} and using \eqref{outlam} and \eqref{yvsize}, 
		\begin{align*}
			C_{0f} = \sup_{n \in \mathbb{N}_0,\mathcal{F} \in {\mathscr{F}(n)}} \frac{\log |\mathcal{M}|}{n+1}  & \leq  \lim_{n \to \infty}\frac{\log \frac{q^{n+1}}{\alpha \lambda^{n+1}}}{n+1}= 1-\log \lambda.
		\end{align*}
	Here, we could replace $ \sup $ with $ \lim  $ due to supperadditivity of the channel.
		This proves the upper bound in \eqref{nsslmd}.
		
		\section{Coding scheme used in the numerical example } \label{app:ne}
			The decoder computes a state estimate $ \hat{X}(t) $ as well as an upper bound $ \delta(t) $ of its exactness. 
			These operations are duplicated at the encoder as well.
			The encoder employs the quantizer $ \mathcal{Q}_n$ with blocklength $n$ and computes the quantized value $ q(\tau_i) $ of the current scaled estimation error at epochs $\tau_i:=ni, i=0,1,\dots,$ produced by the encoder–decoder pair:
			\begin{align*}
			q(\tau_i)=\mathcal{Q}_n(\varepsilon(\tau_i)),\quad\varepsilon(\tau_i) := \frac{X(\tau_i)-\hat{X}(\tau_i)}{\delta(\tau_i)}. 
			\end{align*}
			Next, the encoder encodes it by means of the encoding technique described above with blocklength $ n=15 $ (every third bit is a parity check, i.e. $ b_{k+2}=b_{k+1}\oplus b_{k}, k=1,4,7,..., 13$) and sends it across the channel during the next epoch $ \left[ \tau_i, \tau_{i+1} \right)  $. 
			\\
			At the decoder, the error-less decoding rule is applied to the data received within the previous epoch $ \left[ \tau_{i-1}, \tau_{i} \right) $ and therefore computes the quantized and scaled estimation error $ q(\tau_{i-1}) $.
			Next, the {\em estimate} $ \hat{X} $ and the {\em exactness bound} $ \delta $ are updated:
			\begin{align*}
			\hat{X}(\tau_i) &= a^n \bigg( \hat{X}(\tau_{i-1}) + \delta(\tau_{i-1}) q(\tau_{i-1})\bigg),\\
			\delta(\tau_i) &= \delta(\tau_{i-1}) 2^{(h_{lin}-C_0)n}+\delta_*,
			\end{align*}
			where $  \delta_* > 2^{h_{lin}n} $ is a constant scalar and $ 2^{(h_{lin}-C_0)n}$ is the contraction rate of the quantizer $ \mathcal{Q}_n $.
			The encoder and decoder are given the common and arbitrarily chosen values $ \hat{X}(0)=0 $ and $ \delta(0)=1$.
		\section{Proof of Theorem \ref{thm:c0lec}~($ C_{0} $ lower bound of sliding window erasure channel)} \label{app:swe}  
		In \cite{saberi2018estimation}, a binary erasure channel is introduced and studied that in every consecutive-window of $ w $ bits at most $ d $ erasures can happen. A generalization of this channel is that the alphabet size can be any $ q \in \{2,3, \dots\} $ defined as {\em consecutive-window erasure channel}. The following Proposition gives a lower bound for the zero-error capacity of this channel.
		\begin{prop}\label{swcwl1}
			The zero-error capacity a consecutive-window erasure channel in which erases up to $ d $ symbols in every consecutive-window of $ w $ symbols is lower bounded by
			\begin{align}
			C^{cw}_0 \geq 1-\frac{1}{w}\log_q V^w_{d}(q). \label{shlb}
			\end{align}
		\end{prop}
		\begin{proof}
			Considering every window of $ w $ transmissions, the channel acts as a mapping from input space $ \mathcal{X}^w $ to $ \mathcal{Y}^w $, where $ \mathcal{Y}=\mathcal{X} \cup \{*\} $. Note that this lifted channel is memoryless and therefore, the results for memoryless channel can be applied. 
			
			We use the lower bound for zero-error capacity given by Shannon in \cite{shannon1956zero}. Let $ p=[p(x_1), p(x_2), ... , p(x_{|\mathcal{X}|})] $ and 
			\[\Xi=\{p \in \mathbb{R}^{|\mathcal{X}|} \,|\, \mathbbm{1}^Tp=1, \, p \succeq 0 \}\] be the input probability vector and simplex of probability vectors on the input set $ \mathcal{X} $, respectively. Shannon proved the following lower bound for zero-error capacity
			\begin{equation}\label{lbnd}
			C_{0} \geq -\log\min_{p \in \Xi}\sum_{i,j} \mathbf{A}_{ij}\,p(x_i)p(x_j),
			\end{equation}
			where $ \mathbf{A}_{ij} $ is the $ (i,j) $-th element of the adjacency matrix, $ \mathbf{A} $. These elements are equal to one if $ i $-th and $ j $-th input are adjacent and zero otherwise.
			The minimization in \eqref{lbnd} can be restated as a quadratic optimization problem as follows
			\begin{align}
			\min_{p \in \Xi} \quad & p^T\mathbf{A}\,p. \label{min1}
			\end{align}
			Generally, $ \mathbf{A} $ is indefinite and the problem of \eqref{min1} is not convex.
			Although, uniformly distributed inputs may not yield the minimum in  Shannon's lower bound, \eqref{lbnd}; they still yield an upper bound for the solution of \eqref{min1} and hence a lower bound on it and hence on the zero-error capacity of the channel. 
			
			It is straightforward to see that the sum of each row in the adjacency matrix is $ V^w_{d}(q):= \sum_{i=0}^{d}{w \choose i}(q-1)^i $ and accordingly, \[\sum_{i,j} \mathbf{A}_{ij}=q^w  V^w_d(q).\]
			Hence, considering a uniform distribution on the input space and \eqref{lbnd}, we have
			\begin{align*}
			C^{cw}_0 &\geq -\frac{1}{w}\log_q \bigg(q^w V^w_d(q)  \frac{1}{|\mathcal{X}|^2}\bigg) =
			1-\frac{1}{w}\log_q V^w_d(q). 
			\end{align*}
		\end{proof}
		Moreover, the following Proposition shows that the zero-error capacity of the sliding-window channel is lower bounded by consecutive-window erasure channel.
		\begin{prop} \label{swcwl2}
			The zero-error capacity of a ($ w,d $) sliding-window erasure channel is lower bounded by the zero-error capacity of a consecutive-window erasure channel in which erases up to $ d $ symbols in every consecutive-window of $ w $ symbols.
		\end{prop}
		\begin{proof} 
			We claim that every zero-error code of length (block code of size) $ wK $, where $ K\in \mathbb{N} $ for the consecutive-window erasure channel is a zero-error code for sliding window erasure channel, as well. 
			
			Consider when we have one window, i.e., $ K=1 $, two channels are the same since the output set for consecutive-window erasure channel, $ \mathcal{Y}_{CW} $ is equal to output set of sliding-window erasure channel, $ \mathcal{Y}_{SW} $, i.e., $ \mathcal{Y}_{CW}=\mathcal{Y}_{SW} $. However, when $ K=2 $, the sliding window channel impose more constraints on occurring erasures such that less combination of erasure and received letters can be formed in output sequence; e.g. in the consecutive-window erasure channel, $ 2d $ consecutive erasures can happen in the output. However, in the sliding-window case these outputs cannot be formed. This argument holds for any $ K>2 $ as well and hence $ \mathcal{Y}_{SW} \subset \mathcal{Y}_{CW} $. Therefore, less outputs can be produced by the combinations of erasures and received symbols, makes some inputs distinguish comparing to consecutive-window erasure channel. Therefore, since less adjoint inputs appears, the zero-error code for consecutive-window erasure channel will ensure error-free transmission for sliding window channel.
			By this, we can send information at least with rate of $ C^{cw}_0 $.
		\end{proof}
		From Propositions \ref{swcwl1} and \ref{swcwl2}, the lower bound in \eqref{bounds_sw} is concluded.
		
		\section{Proof of Theorem \ref{thm:nssc0f}~($ C_{0} $ bounds of sliding window symmetric channel)} \label{app:sws}  
		The upper and lower bounds are discussed separately in the following subsections.
		\subsection{Upper bound}
		For deriving an upper bound for the zero-error feedback capacity, we assume that the decoder has access to the state information via a {\em side channel}. In other words, by gifting the state information via a genie aided channel, we use \eqref{c0Fs} to derive an upper bound for the zero-error feedback capacity. However, with the state diagram constructed in section \ref{swcm}, each output leads to a different state. Therefore, by having the current state information, the decoder can determine the previous input, that is $	x(t-1)=g(s(t),y(t-1))$.
		
		Hence, channel will perform error-less and therefore $ C_{0f} =1 $. To derive a tighter upper bound we use another form of state representation for the sliding-window symmetric channel which is similar to the finite-state erasure state diagram. In this model, instead of having a one-to-one correspondence between the outputs and states, each erroneous state has two outgoing edges, one for error-free transmission ($  y(t) = x(t) $) and another one when there is an error ($ y(t) \neq x(t) $).
		
		Note with this assumption and having the state information at the decoder, the rest of the analysis is similar to the finite-state erasure channel (Appendix \ref{app:c}) with this difference that for the non-binary channel, an error can take $ q-1 $ different values.
		
		For this channel, $ \mathcal{G}(y,s'|s)=y $ when the transmission is error-free and $ \mathcal{G}(y,s'|s)=\mathcal{X}\backslash\{y\} $ otherwise. Using the same line of reasoning in Appendix \ref{app:c}, an upper bound on the solution of DP in \eqref{wit} can be derived which is stated in the following lemma.  
		\begin{lem}
			For a $ (w,d) $ sliding-window symmetric channel the solution of \eqref{wit} satisfies
			\begin{align*}
			W(k,s)&\leq\frac{q^k}{( q-1 )^{ k\frac{d}{w}-w-|\mathcal{S}|}}.
			\end{align*}
		\end{lem}
		\begin{proof}
			The states (similar to finite-state erasure channel in Appendix \ref{app:c}) can be partitioned into two subsets $ \mathcal{S}_I $ and $ \mathcal{S}_{II} $ ,where, $ \mathcal{S}_I $ contains states that next action can take them in two states, one error-free transmission and the other one, with error. 
			
			Similar to finite-state erasure channel, when $ s \in \mathcal{S}_I$, we have two possible edges ending in state $s_s$ for error-free transmission and $s_e$ for transmission with error. Thus, the solution of \eqref{wit} $ \forall s \in \mathcal{S}_I$ is as follows	
			\begin{align}
			\begin{split}
			W(k,s)=&\max_{P_{X|S}} \min\bigg\{ W(k-1,s_s)\big(\max_{y \in \mathcal{Y}} P_{X|S}(x=y|s)\big)^{-1}\\
			&\qquad, W(k-1,s_e)\big(\max_{y \in \mathcal{Y}} \sum_{x \in \mathcal{X}\backslash\{y\}} P_{X|S}(x|s)\big)^{-1}\bigg\}
			\end{split}\nonumber\\
			\begin{split}
			\stackrel{(a)}{\leq}& \min\bigg\{ W(k-1,s_s)\max_{P_{X|S}}\big(\max_{y \in \mathcal{Y}} P_{X|S}(x=y|s)\big)^{-1}\\
			& \,,W(k-1,s_e)\max_{P_{X|S}} \big(\max_{y \in \mathcal{Y}} \sum_{x \in \mathcal{X}\backslash\{y\}} P_{X|S}(x|s)\big)^{-1}\bigg\}
			\end{split}\nonumber\\
			\stackrel{(b)}{\leq}&\min \big\{ W(k-1,s_s)\times q, W(k-1,s_e)  \times \frac{q}{q-1} \big\}, \label{iterup}
			\end{align}	
			where, (a) holds since distributing $\max_{P_{X|S}}$ inside the minimization, builds an upper-bound for the max-min problem. Moreover, the uniform distribution is the solution of both elements inside the minimization, specifically \[\max_{P_{X|S}} (\max_{y \in \mathcal{Y}} \sum_{x \in \mathcal{X}\backslash\{y\}} P_{X|S}(x|s))^{-1}=q/(q-1)\] which justifies (b).
			
			If the next state $ s' $ is in $ \mathcal{S}_{II} $,
			\begin{align}
			W(k,s)&= \max_{P_{X|S}} \Bigg\{W(k-1,s') \big(\max_{y \in \mathcal{Y}}P_{X|S}(x=y|s)\big)^{-1} \Bigg\}\nonumber\\
			&=q \times W(k-1,s') ,\; \forall s \in \mathcal{S}_{II} \label{its41}
			\end{align}
			which we have worst-case gain as follows
			\begin{align}
			a(s,s')=
			\begin{cases}
			\frac{q}{q-1} & \text{if } \mathcal{G}(y,s'|s)=\mathcal{X}\backslash\{y\},\\
			q & \text{if } \mathcal{G}(y,s'|s)=y.
			\end{cases}
			\end{align}
			Consider the cycle with maximal ratio for sliding window channel is $ \tau=d/w $ and therefore by same line of reasoning for finite-state erasure channel (Lemma \ref{lem:miner}), there is an output sequence that has at least $ \tau k -l- |\mathcal{S}|$ number of errors (not erasures) where here $ \tau=d/w $ and $ l=w $. Further, since \eqref{iterup} is an upper bound rather that equality we have
			\begin{align*}
			W(k,s&)\leq\min_{s(0:k-1) \in \mathscr{S}(k)} \prod_{i=0}^{k-1} a(s(i),s(i+1)),\text{ s.t. } s(0)=s\\
			&\leq q^{k-\tau k + l + |\mathcal{S}|} \bigg( \frac{q}{q-1} \bigg)^{\tau k-l-|\mathcal{S}|}=\frac{q^k}{( q-1 )^{\tau k-l-|\mathcal{S}|}}.
			\end{align*}	
		\end{proof}
		Since the state information is not available for the original channel, this gives an upper bound for the zero-error feedback capacity. 
		\begin{align*}
		C_{0f}&\leq\lim_{k \rightarrow \infty} \frac{1}{k} \log_q W(k,s) \\
		&\leq 1-\lim_{k \rightarrow \infty} \frac{\log_q( q-1 )^{ k\frac{d}{w}-w-|\mathcal{S}|}}{k}=1-\frac{d}{w}\log_q (q-1).
		\end{align*}	
		
		Next, we show that if $ d \geq w/2 $, $ C_{0f}=0 $. We use the notation used to derive the upper bound for the finite-state additive noise channel in Appendix \ref{app:nsslw}.
		
		Since $ d \geq w/2 $, i.e., the number of errors is equal or larger than the error-free ones,  considering any two encoding function (say $ f_m(0:n-1) $ and $ f_{m'}(0:n-1) $) can cause same output, i.e. \[f_m(0:n-1) \oplus v(0:n-1)= f_{m'}(0:n-1) \oplus v'(0:n-1) .\] In other words, there are enough errors that lead to at least one same output. Therefore, $  C_{0f}=0  $ and thus $ C_0=0 $ for $ d \geq w/2 $.
		\subsection{Lower bound}
		First, we define {\em consecutive-window symmetric channel} which is an extension of consecutive-window erasure channel introduced in Appendix \ref{app:swe}.  The difference here is that instead of erasure an error can happen which is any symbol mapped to any symbol other than sent one. We have the following lower bound on zero-error capacity of this channel.
		\begin{prop}\label{swcwsl1}
			The zero-error capacity of a consecutive-window symmetric channel in which up to $ d $ errors can occur in every consecutive-window of $ w $ symbols is lower bounded by
			\[C^{cw}_0 \geq 1-\frac{1}{w}\log_q V^w_{2d}(q).\]
		\end{prop}
		\begin{proof}
			From coding theory we know that in a block of $ w $ symbols, if the codewords have the Hamming distance of $ 2\times d $, then up to $ d $ errors can be corrected. Therefore, the number of adjacent inputs  for every codeword is equal to $ V_{2d}^w (q) :=1+{w \choose 1}(q-1)+ ... +{w \choose 2d}(q-1)^{2d} $, including the codeword itself. 
			
			According to definition of adjacency, sum of each row is the number of adjacent inputs; therefore, $ V_{2d}^w(q)$ is the sum of rows  of the adjacency matrix. Summing on all elements gives \[\sum_{i,j} \mathbf{A}_{ij}=|\mathcal{X}|V_{2d}^w(q) .\]  Considering a uniform distribution on the input space and \eqref{lbnd} yield
			\begin{align*}
			C_0 &\geq -\frac{1}{w}\log_q\bigg(q^wV_{2d}^w  \frac{1}{|\mathcal{X}|^2}\bigg) =1-\frac{1}{w}\log_qV_{2d}^w(q). 		\end{align*}
		\end{proof}
		Moreover, the following Proposition shows that the zero-error capacity of the sliding-window channel is lower bounded by consecutive-window erasure channel. The proof follows from the same line of reasoning as in Proposition \ref{swcwl2}.
		\begin{prop} \label{swcwsl2}
			The zero-error capacity of a $ (w,d)$ sliding-window symmetric channel is lower bounded by the zero-error capacity of a consecutive-window symmetric channel in which up to $ d $ errors in every consecutive-window of $ w $ symbols.
		\end{prop}
		
		Propositions \ref{swcwsl1} and \ref{swcwsl2} gives the lower bound in \eqref{nssc0f}.
		
		%\begin{thebibliography}{1}
		%
		\bibliographystyle{IEEEtran}
		\bibliography{ieeetranb}	

\begin{thebibliography}{10}
\providecommand{\url}[1]{#1}
\csname url@samestyle\endcsname
\providecommand{\newblock}{\relax}
\providecommand{\bibinfo}[2]{#2}
\providecommand{\BIBentrySTDinterwordspacing}{\spaceskip=0pt\relax}
\providecommand{\BIBentryALTinterwordstretchfactor}{4}
\providecommand{\BIBentryALTinterwordspacing}{\spaceskip=\fontdimen2\font plus
\BIBentryALTinterwordstretchfactor\fontdimen3\font minus
  \fontdimen4\font\relax}
\providecommand{\BIBforeignlanguage}[2]{{%
\expandafter\ifx\csname l@#1\endcsname\relax
\typeout{** WARNING: IEEEtran.bst: No hyphenation pattern has been}%
\typeout{** loaded for the language `#1'. Using the pattern for}%
\typeout{** the default language instead.}%
\else
\language=\csname l@#1\endcsname
\fi
#2}}
\providecommand{\BIBdecl}{\relax}
\BIBdecl

\bibitem{saberi2019state}
A.~Saberi, F.~Farokhi, and G.~N. Nair, ``State estimation via worst-case
  erasure and symmetric channels with memory,'' in \emph{IEEE International
  Symposium on Information Theory (ISIT)}, 2019, pp. 3072--3076.

\bibitem{saberi2020explicit}
A.~{Saberi}, F.~{Farokhi}, and G.~N. {Nair}, ``An explicit formula for the
  zero-error feedback capacity of a class of finite-state additive noise
  channels,'' in \emph{IEEE International Symposium on Information Theory
  (ISIT)}, 2020, pp. 2108--2113.

\bibitem{haddadin2017robot}
S.~Haddadin, A.~De~Luca, and A.~Albu-Sch{\"a}ffer, ``Robot collisions: A survey
  on detection, isolation, and identification,'' \emph{IEEE Transactions on
  Robotics}, vol.~33, no.~6, pp. 1292--1312, 2017.

\bibitem{teixeira2012attack}
A.~Teixeira, D.~P{\'e}rez, H.~Sandberg, and K.~H. Johansson, ``Attack models
  and scenarios for networked control systems,'' in \emph{Proceedings of the
  1st International Conference on High Confidence Networked Systems}.\hskip 1em
  plus 0.5em minus 0.4em\relax ACM, 2012, pp. 55--64.

\bibitem{linsenmayer2017stabilization}
S.~Linsenmayer and F.~Allgower, ``Stabilization of networked control systems
  with weakly hard real-time dropout description,'' in \emph{Decision and
  Control (CDC), 2017 IEEE 56th Annual Conference on}, 2017, pp. 4765--4770.

\bibitem{schenato2007foundations}
L.~Schenato, B.~Sinopoli, M.~Franceschetti, K.~Poolla, and S.~S. Sastry,
  ``Foundations of control and estimation over lossy networks,''
  \emph{Proceedings of the IEEE}, vol.~95, no.~1, pp. 163--187, 2007.

\bibitem{matveev2009estimation}
A.~S. Matveev and A.~V. Savkin, \emph{Estimation and control over communication
  networks}.\hskip 1em plus 0.5em minus 0.4em\relax Springer Science \&
  Business Media, 2009.

\bibitem{sukhavasi2016linear}
R.~T. Sukhavasi and B.~Hassibi, ``Linear time-invariant anytime codes for
  control over noisy channels,'' \emph{IEEE Transactions on Automatic Control},
  vol.~61, no.~12, pp. 3826--3841, 2016.

\bibitem{bernat2001weakly}
G.~Bernat, A.~Burns, and A.~Liamosi, ``Weakly hard real-time systems,''
  \emph{IEEE Transactions on Computers}, vol.~50, no.~4, pp. 308--321, 2001.

\bibitem{xiong2007stabilization}
J.~Xiong and J.~Lam, ``Stabilization of linear systems over networks with
  bounded packet loss,'' \emph{Automatica}, vol.~43, no.~1, pp. 80--87, 2007.

\bibitem{badr2017layered}
A.~Badr, P.~Patil, A.~Khisti, W.-T. Tan, and J.~Apostolopoulos, ``Layered
  constructions for low-delay streaming codes,'' \emph{IEEE Transactions on
  Information Theory}, vol.~63, no.~1, pp. 111--141, 2017.

\bibitem{fong2019optimal}
S.~L. Fong, A.~Khisti, B.~Li, W.-T. Tan, X.~Zhu, and J.~Apostolopoulos,
  ``Optimal streaming codes for channels with burst and arbitrary erasures,''
  \emph{IEEE Transactions on Information Theory}, 2019.

\bibitem{avestimehr2011wireless}
A.~S. Avestimehr, S.~N. Diggavi, and D.~Tse, ``Wireless network information
  flow: A deterministic approach,'' \emph{IEEE Transactions on Information
  theory}, vol.~57, no.~4, pp. 1872--1905, 2011.

\bibitem{bandi2012tractable}
C.~Bandi and D.~Bertsimas, ``Tractable stochastic analysis in high dimensions
  via robust optimization,'' \emph{Mathematical Programming}, vol. 134, no.~1,
  pp. 23--70, 2012.

\bibitem{matveev2007shannon}
A.~S. Matveev and A.~V. Savkin, ``Shannon zero error capacity in the problems
  of state estimation and stabilization via noisy communication channels,''
  \emph{International Journal of Control}, vol.~80, pp. 241--255, 2007.

\bibitem{Franceschetti2014}
M.~Franceschetti and P.~Minero, ``Elements of information theory for networked
  control systems,'' in \emph{Information and Control in Networks}, G.~Como,
  B.~Bernhardsson, and A.~Rantzer, Eds.\hskip 1em plus 0.5em minus 0.4em\relax
  Cham: Springer International Publishing, 2014, pp. 3--37.

\bibitem{wiese2018secure}
M.~Wiese, T.~J. Oechtering, K.~H. Johansson, P.~Papadimitratos, H.~Sandberg,
  and M.~Skoglund, ``Secure estimation and zero-error secrecy capacity,''
  \emph{IEEE Transactions on Automatic Control}, 2018.

\bibitem{shannon1956zero}
C.~Shannon, ``The zero error capacity of a noisy channel,'' \emph{IRE
  Transactions on Information Theory}, vol.~2, no.~3, pp. 8--19, 1956.

\bibitem{korner1998zero}
J.~Korner and A.~Orlitsky, ``Zero-error information theory,'' \emph{IEEE
  Transactions on Information Theory}, vol.~44, no.~6, pp. 2207--2229, 1998.

\bibitem{badr2017perfecting}
A.~Badr, A.~Khisti, W.-T. Tan, and J.~Apolstolopoulos, ``Perfecting protection
  for interactive multimedia: A survey of forward errror correction for
  low-delay interactive applications,'' \emph{IEEE Signal Processing Magazine},
  vol.~34, no.~2, pp. 95--113, 2017.

\bibitem{wang2018end}
Q.~Wang and S.~Jaggi, ``End-to-end error-correcting codes on networks with
  worst-case bit errors,'' \emph{IEEE Transactions on Information Theory},
  vol.~64, no.~6, pp. 4467--4479, 2018.

\bibitem{nair2013nonstochastic}
G.~N. Nair, ``A nonstochastic information theory for communication and state
  estimation,'' \emph{IEEE Transactions on Automatic Control}, vol.~58, no.~6,
  pp. 1497--1510, 2013.

\bibitem{zhu2019observer}
Y.~Zhu and W.~X. Zheng, ``Observer-based control for cyber-physical systems
  with periodic dos attacks via a cyclic switching strategy,'' \emph{IEEE
  Transactions on Automatic Control}, 2019.

\bibitem{zhu2019quasi}
Y.~Zhu, W.~X. Zheng, and D.~Zhou, ``Quasi-synchronization of discrete-time
  lur’e-type switched systems with parameter mismatches and relaxed pdt
  constraints,'' \emph{IEEE Transactions on Cybernetics}, vol.~50, no.~5, pp.
  2026--2037, 2019.

\bibitem{middleton2009feedback}
R.~H. Middleton, A.~J. Rojas, J.~S. Freudenberg, and J.~H. Braslavsky,
  ``Feedback stabilization over a first order moving average gaussian noise
  channel,'' \emph{IEEE Transactions on Automatic Control}, vol.~54, no.~1, pp.
  163--167, 2009.

\bibitem{goldsmith1996capacity}
A.~J. Goldsmith and P.~P. Varaiya, ``Capacity, mutual information, and coding
  for finite-state markov channels,'' \emph{IEEE Transactions on Information
  Theory}, vol.~42, no.~3, pp. 868--886, 1996.

\bibitem{han2015randomized}
G.~Han, ``A randomized algorithm for the capacity of finite-state channels,''
  \emph{IEEE Transactions on Information Theory}, vol.~61, no.~7, pp.
  3651--3669, 2015.

\bibitem{kovavcevic2017zero}
M.~Kova{\v{c}}evi{\'c}, M.~Stojakovi{\'c}, and V.~Y. Tan, ``Zero-error capacity
  of $ p $-ary shift channels and fifo queues,'' \emph{IEEE Transactions on
  Information Theory}, vol.~63, no.~12, pp. 7698--7707, 2017.

\bibitem{cao2018zero}
Q.~Cao, N.~Cai, W.~Guo, and R.~W. Yeung, ``On zero-error capacity of binary
  channels with one memory,'' \emph{IEEE Transactions on Information Theory},
  vol.~64, no.~10, pp. 6771--6778, 2018.

\bibitem{zhao2010zero}
L.~Zhao and H.~H. Permuter, ``Zero-error feedback capacity of channels with
  state information via dynamic programming,'' \emph{IEEE Transactions on
  Information Theory}, vol.~56, no.~6, pp. 2640--2650, 2010.

\bibitem{fong2018optimal}
S.~L. Fong, A.~Khisti, B.~Li, W.-T. Tan, X.~Zhu, and J.~Apostolopoulos,
  ``Optimal streaming erasure codes over the three-node relay network,''
  \emph{arXiv preprint arXiv:1806.09768}, 2018.

\bibitem{adler1965topological}
R.~L. Adler, A.~G. Konheim, and M.~H. McAndrew, ``Topological entropy,''
  \emph{Transactions of the American Mathematical Society}, vol. 114, no.~2,
  pp. 309--319, 1965.

\bibitem{downarowicz2011entropy}
T.~Downarowicz, \emph{Entropy in dynamical systems}.\hskip 1em plus 0.5em minus
  0.4em\relax Cambridge University Press, 2011, vol.~18.

\bibitem{kawan2019metric}
C.~Kawan and S.~Yuksel, ``Metric and topological entropy bounds for optimal
  coding of stochastic dynamical systems,'' \emph{IEEE Transactions on
  Automatic Control}, 2019.

\bibitem{liberzon2017entropy}
D.~Liberzon and S.~Mitra, ``Entropy and minimal bit rates for state estimation
  and model detection,'' \emph{IEEE Transactions on Automatic Control},
  vol.~63, no.~10, pp. 3330--3344, 2017.

\bibitem{lind1995introduction}
D.~Lind and B.~Marcus, \emph{An introduction to symbolic dynamics and
  coding}.\hskip 1em plus 0.5em minus 0.4em\relax Cambridge University Press,
  1995.

\bibitem{saberi2018estimation}
A.~Saberi, F.~Farokhi, and G.~Nair, ``Estimation and control over a
  nonstochastic binary erasure channel,'' in \emph{7th IFAC Workshop on
  Distributed Estimation and Control in Networked Systems (NecSys18)}, vol.~51,
  no.~23.\hskip 1em plus 0.5em minus 0.4em\relax Elsevier, 2018, pp. 265--270.

\bibitem{kellett2005robustness}
C.~M. Kellett and A.~R. Teel, ``On the robustness of kl-stability for
  difference inclusions: Smooth discrete-time lyapunov functions,'' \emph{SIAM
  Journal on Control and Optimization}, vol.~44, no.~3, pp. 777--800, 2005.

\bibitem{brualdi2008combinatorial}
R.~A. Brualdi and D.~Cvetkovic, \emph{A combinatorial approach to matrix theory
  and its applications}.\hskip 1em plus 0.5em minus 0.4em\relax Chapman and
  Hall/CRC, 2008.

\bibitem{sabag2016single}
O.~Sabag, H.~H. Permuter, and H.~D. Pfister, ``A single-letter upper bound on
  the feedback capacity of unifilar finite-state channels,'' \emph{IEEE
  Transactions on Information Theory}, vol.~63, pp. 1392--1409, 2016.

\bibitem{holcombe1982algebraic}
M.~Holcombe and W.~Holcombe, \emph{Algebraic automata theory}.\hskip 1em plus
  0.5em minus 0.4em\relax Cambridge University Press, 1982, vol.~1.

\bibitem{renyi1970foundations}
A.~R{\'e}nyi, \emph{Foundations of probability}.\hskip 1em plus 0.5em minus
  0.4em\relax Holden-Day, 1970.

\bibitem{berman1994nonnegative}
A.~Berman and R.~J. Plemmons, \emph{Nonnegative matrices in the mathematical
  sciences}.\hskip 1em plus 0.5em minus 0.4em\relax SIAM, 1994, vol.~9.

\bibitem{huffman2010fundamentals}
W.~C. Huffman and V.~Pless, \emph{Fundamentals of error-correcting
  codes}.\hskip 1em plus 0.5em minus 0.4em\relax Cambridge university press,
  2010.

\bibitem{roger1994topics}
R.~A. Horn and C.~R. Johnson, \emph{Topics in matrix analysis}.\hskip 1em plus
  0.5em minus 0.4em\relax Cambridge University Press, 1994.

\bibitem{tatikonda2004control2}
S.~Tatikonda and S.~Mitter, ``Control under communication constraints,''
  \emph{IEEE Transactions on Automatic Control}, vol.~49, pp. 1056--1068, 2004.

\bibitem{godsil2013algebraic}
C.~Godsil and G.~F. Royle, \emph{Algebraic graph theory}.\hskip 1em plus 0.5em
  minus 0.4em\relax Springer Science \& Business Media, 2013, vol. 207.

\bibitem{ahlswede1993nonbinary}
R.~Ahlswede, L.~A. Bassalygo, and M.~S. Pinsker, ``Nonbinary codes correcting
  localized errors,'' \emph{IEEE Transactions on Information Theory}, vol.~39,
  no.~4, pp. 1413--1416, 1993.

\bibitem{ahlswede2006non}
R.~Ahlswede, C.~Deppe, and V.~Lebedev, ``Non-binary error correcting codes with
  noiseless feedback, localized errors, or both,'' in \emph{IEEE International
  Symposium on Information Theory}, 2006, pp. 2486--2487.

\end{thebibliography}
\end{document}